\documentclass[journal,12pt,onecolumn, draftclsnofoot]{IEEEtran}

\usepackage[english]{babel}
\usepackage{csquotes}

\usepackage{amsmath,amssymb,amsfonts,amsthm}
\usepackage{bbm}
\usepackage{algorithmic}
\usepackage{graphicx}
\usepackage{textcomp}
\usepackage{xcolor}
\usepackage{hyperref}
\usepackage{pgf}
\def\BibTeX{{\rm B\kern-.05em{\sc i\kern-.025em b}\kern-.08em
    T\kern-.1667em\lower.7ex\hbox{E}\kern-.125emX}}
\usepackage{cite}
\usepackage{csquotes}
\MakeOuterQuote{"}
\usepackage{microtype}

\newtheorem{theorem}{Theorem}
\newtheorem{lemma}{Lemma}
\newtheorem{corollary}{Corollary}

\theoremstyle{definition}
\newtheorem{definition}{Definition}

\theoremstyle{remark}
\newtheorem{remark}{Remark}

\newcommand{\indep}{\mathrel{\bot}\joinrel\mathrel{\mkern-5mu}%
	\joinrel\mathrel{\bot}}

\newcommand{\hbtoterror}{\epsilon^{({n_1,n_2})}}
\newcommand{\hbinderror}[1]{\epsilon_{#1}^{({n_{#1}})}}
\newcommand{\hbanderror}{\epsilon_{1 \land 2}^{({n_1,n_2})}}
\newcommand{\hbtoterrortilde}{\tilde{\epsilon}^{({n_1,n_2})}}
\newcommand{\hbinderrortilde}[1]{\tilde{\epsilon}_{#1}^{({n_{#1}})}}
\newcommand{\hbanderrortilde}{\tilde{\epsilon}_{1 \land 2}^{({n_1,n_2})}}

\newcommand{\firstpart}{\mathrm{I}}
\newcommand{\secondpart}{\mathrm{II}}
\newcommand{\unitmatrix}[1]{\mathbf{\mathrm{I}}^{#1\times #1}}

\begin{document}

\title{New Inner and Outer Bounds for Gaussian Broadcast Channels with Heterogeneous Blocklength Constraints}

\author{\IEEEauthorblockN{Marcel Mross, Pin-Hsun Lin and Eduard A. Jorswieck}

\IEEEauthorblockA{Institute for Communications Technology, Technische Universität Braunschweig, Germany \\
E-mail: \{Mross, Lin, Jorswieck\}@ifn.ing.tu-bs.de
}
}

\maketitle
\begin{abstract}
We investigate novel inner and outer bounds on the rate region of a 2-user Gaussian broadcast channel with finite, heterogeneous blocklength constraints (HB-GBC). In particular, we introduce a new, modified Sato-type outer bound that can be applied in the finite blocklength regime and which does not require the same marginal property. We then develop and analyze \textit{composite shell codes}, which are suitable for the HB-GBC. Especially, to achieve a lower decoding latency for the user with a shorter blocklength constraint when successive interference cancellation is used, we derive the number of symbols needed to successfully early decode the other user's message. We numerically compare our derived outer bound to the best known achievable rate regions. Numerical results show that the new early decoding performance in terms of latency reduction is significantly improved compared to the state of the art, and it performs very close to the asymptotic limit.
\end{abstract}


\section{Introduction}
As one of the application areas of 5G and beyond, Ultra-Reliable Low Latency Communication (URLLC) has attracted intensive attention. To achieve low latency, asymptotic assumptions on the codeword sizes are no longer valid, which motivates the finite blocklength analysis. Since the seminal work in \cite{Polyanskiy.2010}, there have been several extensions to multi-user channels, like the multiple access channel (MAC) \cite{MolavianJazi.2015}, the MAC with degraded message sets \cite{Scarlett.2015}, the broadcast channel \cite{Tan.2014, Unsal.2017, Sheldon.2021} and the interference channel \cite{Scarlett.2017}.

Due to different demands on Quality of Service for the wide range of available applications in modern communication systems, the latency requirements can differ among users. Therefore, in a realistic model, the two users may be required to decode after having received differing numbers of symbols $n_1$ and $n_2$, which leads to the model of the Gaussian broadcast channel with heterogeneous blocklength constraints (HB-GBC) \cite{Tuninetti.2018, Xu.2020, Lin.2021, Lin.2021b, Lin.2022}. The work \cite{Lin.2021} shows the possibility to perform superposition coding at the encoder and successive interference cancellation (SIC) at the decoder. The sufficient conditions for a successful SIC are: 1) the channel gain of the user with the stricter blocklength constraint is larger than that of the other user and 2) the shorter blocklength between the two codewords is still sufficiently large. This technique is called \textit{early decoding}, since the longer codeword is decoded before it is entirely received at the non-intended receiver. How early the decoder can successfully decode was derived for asymptotic cases \cite{Azarian.2005} and as a second-order finite blocklength result \cite{Lin.2021}. 

However, it is unclear how the heterogeneous blocklength assumption affects the capacity and, thus, how much the rates derived in \cite{Lin.2021b, Lin.2022} can be improved. Sato's outer bound \cite{Sato.1978}, a well-known asymptotic outer bound for broadcast channels, requires the same marginal property (SMP), which is invalid in the finite blocklength regime. Also, the authors of \cite{Lin.2021, Lin.2021b, Lin.2022} only consider independent and identically distributed (i.i.d.) Gaussian codebooks, because shell codebooks bring unique challenges to the analysis in the heterogeneous blocklength scenario. Since shell codes are known to have a smaller dispersion than i.i.d. Gaussian codebooks, it is highly desirable to apply them in early decoding.

The main contributions of this paper are as follows. 
\begin{itemize}
    \item  First, we present a novel scheme of applying a Sato-type outer bound technique in the finite blocklength regime, where the SMP is invalid in general. Our scheme still allows the minimization of the outer bound with respect to the joint distributions having the same marginals. We then apply this idea to the HB-GBC, leading to the first existing outer bound for this channel model. We combine the bounding technique with the information spectrum converse method \cite{Hayashi.2009} to derive the outer bounds.
    \item Second, we improve the achievable rate and the latency of the ED scheme proposed in \cite{Lin.2021} for the considered model. For this purpose, we introduce a new class of shell codes, namely, \textit{composite shell codes}, which is suitable for the heterogeneous blocklength scenario and leads to a second-order latency performance close to the asymptotic limit. The derived outer bounds are numerically compared to the best known achievable rate regions. 
    \item Additionally, we investigate some other fundamental properties of the GBC in the finite blocklength regime, which are different to those in the asymptotic case.
\end{itemize} 

This paper is organized as follows. In Section \ref{sec:preliminaries}, we introduce the system model and some preliminaries. In Section \ref{sec:outer_bound}, we present our new outer bound technique and the application to the GBC with heterogeneous blocklength constraints. In Section \ref{sec:achievable}, we present the improved achievability results. The properties of the GBC in the finite blocklength regime are discussed in Section \ref{sec:properties}. Section \ref{sec:numerical} shows some numerical calculations of the results from the previous sections.

\section{Preliminaries and System Model}\label{sec:preliminaries}
\subsection{Notation}
Real constants are denoted by uppercase sans serif letters like $\mathsf{X}$ and $\mathsf{Y}$. Vectors are written with a superscript that indicates the dimension of that vector, e.g. $x^n = [x_1, ..., x_n]$. The individual elements of that vector are indexed by subscripts, i.e., $x_i$. If the vector itself already has a subscript, e.g. $x_1^n$, then the index follows that subscript, separated by a comma: $x_{1,i}$. The concatenation of two vectors $x^n$ and $y^m$ is written as $[x^n, y^m]$ and the vector that consists solely of zeros is denoted by $\mathbf{0}^n$. Sets are denoted by calligraphic letters like $\mathcal{X}$ and $\mathcal{Y}$.

Random variables are denoted by uppercase letters like $X$ and $Y$, realizations of random variables are written in lowercase, like $x$ and $y$. We write $X \sim P_X$ to indicate that the random variable $X$ follows the distribution $P_X$, where $P_X$ denotes a probability measure. The expectation and the variance of a random variable $X$ are denoted by $\mathbb{E}[X]$ and $\mathrm{Var}[X]$, respectively, and the covariance of $X$ and $Y$ is denoted by $\mathrm{Cov}[X,Y]$. If two random variables $X$ and $Y$ are independent, we write $X \indep Y$. A random vector is indicated by a superscript like a deterministic vector. 
We will denote the Gaussian distribution with mean $\mu$ and variance $\sigma^2$ by $\mathcal{N}(\mu, \sigma^2)$ and the $Q$-function by $Q(\cdot)$. We will also use the indicator function $\mathbbm{1}(\cdot)$ and the Landau symbols $o(\cdot)$ and $O(\cdot)$. We also denote the Gaussian capacity function by $\mathsf{C}(x) = \frac{1}{2} \cdot \log(1+x)$, the shell dispersion by $\mathsf{V}(x) := \frac{\log^2 e}{2}\frac{x(x+2)}{(x+1)^2}$, and the i.i.d. Gaussian dispersion by $\mathsf{V}_{\mathrm{G}}(x) := \log^2 e\frac{x}{x+1}$. All logarithms are taken to the base 2.

\subsection{System Model}
We consider a two-user Gaussian broadcast channel with heterogeneous blocklength constraints (HB-GBC) and quasi-static block flat-fading. The received signal at user $k$ at time $i \in \{1, ..., n_k\}$ is
\begin{align}
	Y_{k,i} = \sqrt{h_k}X_{i}+ Z_{k,i}, \label{eq:system_model}
\end{align}
where $k = 1,2$ and $Z_{1,i} \sim \mathcal{N}(0,1)$ and $Z_{2,i} \sim \mathcal{N}(0,1)$ are i.i.d. and mutually independent additive white Gaussian noise random variables. The transmitter as well as the receivers have perfect knowledge of the  channel gains $h_k$. Without loss of generality, in this paper we will always assume $n_1 \geq n_2$. Furthermore, we will consider the case that the shorter blocklength constraint belongs to the user with the larger channel gain, i.e., $h_2 \geq h_1$.

\begin{definition}
	An $(n_1, n_2, \mathsf{M}_1, \mathsf{M}_2, \epsilon, \mathcal{F}^{n_1})$-code for an HB-GBC $W$ consists of:
	\begin{itemize}
		\item two message sets $\mathcal{M}_k = \{1, ..., \mathsf{M}_k\}$, $k = 1, 2$,
		\item an encoder $f_{n_1}: \mathcal{M}_1 \times \mathcal{M}_2 \to \mathcal{F}^{n_1}$, where $\mathcal{F}^{n_1} \subseteq \mathbb{R}^{n_1}$ is some pre-defined set of feasible codewords,
		\item two decoders $\phi_{n_k}: \mathbb{R}^{n_k} \to \mathcal{M}_k$, $k=1,2$,
	\end{itemize}
	such that the average system error probability satisfies
	\begin{IEEEeqnarray}{rCl}
		\hbtoterror &:=& \Pr \left[\hat{M}_1 \neq M_1 \text{ or } \hat{M}_2 \neq M_2\right] \leq \epsilon. \label{eq:total_error_prob}
	\end{IEEEeqnarray}
\end{definition}

We denote by $\mathcal{C}_{W}(n_1, n_2, \mathcal{F}^{n_1}, \epsilon)$ the set of all second-order achievable message size pairs for an HB-GBC $W$, i.e., pairs of the form $\log (\mathsf{M}_k) = n_k\mathsf{C}_k + \sqrt{n_k\mathsf{V}_k} + o\left(\sqrt{n_k}\right), \quad k = 1,2,$ with some positive constants $\mathsf{C}_k$ and $\mathsf{V}_k$ for which an $(n_1, n_2, \mathsf{M}_1, \mathsf{M}_2, \epsilon, \mathcal{F}^{n_1})$-code exists.

It is common to assume a power constraint $\mathsf{P}$ on the codewords. In that case, the feasible set of channel inputs is 
\begin{IEEEeqnarray}{rCl}
	\mathcal{F}^{n_1} &=& \mathcal{F}_{\mathrm{max}}^{n_1}(\mathsf{P}) := \left\{x^{n_1}: \|x^{n_1}\|^2 \leq n_1\mathsf{P} \right\}. \label{eq:max_power_constraint}
\end{IEEEeqnarray}
When \textit{superposition coding} is used at the encoder, the codewords $x_1^{n_1}$ and $x_2^{n_2}$ are generated independently for user 1 and user 2 and then superimposed to get the channel input $x^{n_1} = x_1^{n_1} + [x_2^{n_2}, \mathbf{0}^{n_1-n_2}]$, where $\mathbf{0}^{n_1-n_2}$ is an $(n_1-n_2)$-dimensional vector consisting solely of zeros. In that case, we also call \eqref{eq:max_power_constraint} \textit{sum power constraint} (SPC). Alternatively, we can also impose a power constraint on the individual codewords $x_1^{n_1}$ and $x_2^{n_2}$. In that case, which we call the \textit{individual power constraint} (IPC), the feasible set of codewords is defined as
\begin{IEEEeqnarray}{rCl}
	\mathcal{F}^{n_1} = \mathcal{F}_{\mathrm{max}}^{n_1, n_2}(\mathsf{P}_1, \mathsf{P}_2) := &&\left\{x^{n_1} = x_1^{n_1} + [x_2^{n_2}, \mathbf{0}^{n_1-n_2}]: \|x_k^{n_k}\|^2 \leq n_k\mathsf{P}_k, \; k = 1,2 \right\}. \label{eq:ipc}\IEEEeqnarraynumspace
\end{IEEEeqnarray}
In this paper, we will use individual error probabilities  $\epsilon_1$ at user 1 and $\epsilon_{\mathrm{SIC},1}$ and $\epsilon_{\mathrm{SIC},2}$ for the different decoding steps of SIC at user 2. In order to conform with the constraint in \eqref{eq:total_error_prob}, we have to choose these individual error probabilities properly s.t. \begin{IEEEeqnarray}{rCl}
    \epsilon_1+\epsilon_{\mathrm{SIC},1}+\epsilon_{\mathrm{SIC},2} &\leq& \epsilon. \label{eq:error_allocation}
\end{IEEEeqnarray}

If user 2 can decode user 1's message $m_1$ from the first $n_2$ received symbols $Y_{2,1}, ..., Y_{2, n_2}$ while fulfilling \eqref{eq:total_error_prob}, we call that a \emph{successful early decoding} \cite{Lin.2021}.

Throughout this paper, we will use the symbol $p := n_2/n_1\in (0,1]$ for the blocklength ratio and $\bar{p} := 1-p$.

\section{A Modified Sato-Type Outer Bound}\label{sec:outer_bound}
In the following, we use $\mathcal{W}_{\mathrm{SM}}(W)$ to denote the set of channels having the same conditional marginals $P_{Y_1|X}(y_1|x)$ and $P_{Y_2|X}(y_2|x)$ as a GBC $W= P_{Y_1Y_2|X}(y_1, y_2|x)$.

\subsection{The Same Marginal Property}\label{sec:smp}
Sato's outer bound \cite{Sato.1978} relies on the SMP, i.e., the property that the capacity regions of Gaussian broadcast channels depend only on the marginal distributions of the channels \cite[Lemma 5.1]{ElGamal.2011}. The SMP only holds as long as the decoding error probabilities at all the receivers are vanishing. In the finite blocklength regime, however, we deal with non-vanishing error probabilities. For the HB-GBC, the average system error probability in \eqref{eq:total_error_prob} decomposes into 
\begin{IEEEeqnarray}{rCl}
    \hbtoterror &=& \hbinderror{1} + \hbinderror{2} - \hbanderror, \label{eq:smp_error_prob_exact}
\end{IEEEeqnarray}
where $\hbinderror{k} := \Pr[\hat{M}_k \neq M_k]$, $k = 1,2$, and $\hbanderror := \Pr[\hat{M}_1 \neq M_1 \text{ and } \hat{M}_2 \neq M_2]$. Since a channel $V \in \mathcal{W}_{\mathrm{SM}}(W)$ can have different $\hbanderror$ from $W$ under otherwise fixed conditions, $\mathcal{C}_{W}(n_1, n_2, \mathcal{F}^{n_1}, \epsilon)$ and $\mathcal{C}_{V}(n_1, n_2, \mathcal{F}^{n_1}, \epsilon)$ can also differ. For example, consider a channel $W$ that has independent noise terms and another channel $V \in \mathcal{W}_{\mathrm{SM}}(W)$ that has highly correlated noise terms. Then we expect $V$ to have a larger $\hbanderror$ than $W$ and as a consequence, a smaller $\hbtoterror$. As a result, we cannot use the SMP in the finite blocklength regime and as a consequence, neither the original Sato-type outer bound.

\subsection{A Modified Sato-Type Outer Bound}
Even though the SMP is not valid in the finite blocklength regime, $\mathcal{C}_{V}(n_1, n_2, \mathcal{F}^{n_1}, \epsilon)$ for a channel $V \in \mathcal{W}_{\mathrm{SM}}(W)$ can serve as an outer bound of that of $W$ if the error probability constraint for $V$ is relaxed from $\epsilon$ to $2\epsilon$. This is the statement of the following result, which can be used instead of the SMP for finite blocklength outer bounds.
\begin{lemma}\label{lemma:outer_bound_lemma}
	Let $W$ be an HB-GBC with blocklength constraints $n_1$ and $n_2$, $n_1 \geq n_2$, a set of feasible codewords $\mathcal{F}^{n_1}$, and average system error probability constraint $\epsilon < 0.25$.  
	Then for all $V \in \mathcal{W}_{\mathrm{SM}}(W)$, we have
	\begin{IEEEeqnarray}{rCl}
		\mathcal{C}_{W}(n_1, n_2, \mathcal{F}^{n_1}, \epsilon) &\subseteq& \mathcal{C}_{V}(n_1, n_2, \mathcal{F}^{n_1}, 2\epsilon).\label{eq:outer_bound_lemma}
	\end{IEEEeqnarray}
\end{lemma}
Please refer to Appendix \ref{sec:outer_bound_lemma_proof} for the proof. We now apply Lemma \ref{lemma:outer_bound_lemma} to the HB-GBC with SPC to derive our main result as follows. 

\begin{theorem}[Sato-type Outer Bound with Heterogeneous Blocklengths]\label{theorem:heterogeneous_sato_outer_bound}
	For a two-user HB-GBC with feasible set $\mathcal{F}_{\mathrm{max}}^{n_1}(\mathsf{P})$ and channel gains $h_2 \geq h_1$, the message sizes of any coding scheme have to fulfill the following inequalities:
	\begin{IEEEeqnarray}{rCl}
	\log \mathsf{M}_k &\leq& n_k\mathsf{C}(h_k\mathsf{P}) - \sqrt{n_k\mathsf{V}(h_k \mathsf{P})}Q^{-1}(\epsilon) + \frac{1}{2}\log n_k+ O\left(1\right), \quad k = 1,2, \label{eq:m1_ub}\\
	\log \mathsf{M}_1 + \log \mathsf{M}_2 &\leq& n_1\mathsf{C}_{\mathrm{s}}^*(h_1, h_2, p, \mathsf{P}) - \sqrt{n_1\mathsf{V}_{\mathrm{s}}^*(h_1, h_2, p, \mathsf{P})}Q^{-1}(2\epsilon) + \frac{1}{2}\log n_1 + O\left(1\right), \label{eq:sum_rate_ub}\IEEEeqnarraynumspace
	\end{IEEEeqnarray}
	where
	\begin{IEEEeqnarray}{rCl}
		\mathsf{C}_{\mathrm{s}}^*(h_1, h_2, p, \mathsf{P}) &=& p\mathsf{C}\left(h_2\mathsf{P}\right)  + \bar{p}\mathsf{C}\left(h_1\mathsf{P}\right) + \log e\frac{\bar{p}}{2}\left(\frac{h_2}{1+h_2\mathsf{P}} - 		\frac{h_1}{1+h_1\mathsf{P}}\right)\mathsf{P}, \label{eq:converse_sato_fo_term_opt} \\
		\mathsf{V}_{\mathrm{s}}^*(h_1, h_2, p, \mathsf{P}) &=& \frac{\log^2 e}{4}\left(\frac{p\cdot2h_{2}^2\mathsf{P}^2+4h_2\mathsf{P}}{(1 + h_{2}\mathsf{P})^2} \!+\! \bar{p}\frac{2h_1^2\mathsf{P}^2}{(1+h_1\mathsf{P})^2} + \bar{p}\,\mathsf{V}_{\mathrm{s,a}}^*(h_1, h_2, \mathsf{P})\right),\label{eq:converse_sato_variance_opt}\\
		\mathsf{V}_{\mathrm{s,a}}^*(h_1, h_2, \mathsf{P}) &=& \begin{cases}
			\frac{4h_{1}\mathsf{P}}{(1+h_{1}\mathsf{P})^2} - \frac{4h_{2}\mathsf{P}}{(1+h_{2}\mathsf{P})^2} & \text{if }\mathsf{P}^2 < \frac{1}{h_1h_{2}}, \\
			0  & \text{otherwise}.
		\end{cases}\label{eq:sato_additional_variance_opt}\IEEEeqnarraynumspace
	\end{IEEEeqnarray}	
\end{theorem}

\begin{proof}
 Appendix \ref{sec:proof_het_sato}.
\end{proof}
We can specialize Theorem \ref{theorem:heterogeneous_sato_outer_bound} to the homogeneous blocklength case by increasing $n_2$ to $n_1$, i.e., inserting $p=1$ into \eqref{eq:sum_rate_ub}. Since increasing $n_2$ gives the cooperative receiver used in the Sato-type outer bound an additional advantage, the result is still an outer bound to the heterogeneous blocklength case.

\begin{corollary}[Sato-type outer bound with Homogeneous Blocklengths]\label{theorem:homogeneous_sato_outer_bound}
	For a two-user HB-GBC with feasible set $\mathcal{F}_{\mathrm{max}}^{n_1}(\mathsf{P})$ and channel gains $h_2 \geq h_1$, the message sizes of any coding scheme have to fulfill the inequalities \eqref{eq:m1_ub} and
	\begin{IEEEeqnarray}{rCl}
		\log \mathsf{M}_1 + \log \mathsf{M}_2 &\leq& n_1\mathsf{C}\left(h_2\mathsf{P}\right) - \sqrt{n_1\mathsf{V}\left(h_2\mathsf{P}\right)}Q^{-1}(2\epsilon) + \frac{1}{2}\log n_1 + O\left(1\right). \label{eq:sum_rate_ub_hom}
	\end{IEEEeqnarray}
\end{corollary}

Comparing only the first-order terms of Theorem \ref{theorem:heterogeneous_sato_outer_bound}  and Corollary \ref{theorem:homogeneous_sato_outer_bound} we find that to get from $\mathsf{C}(h_2\mathsf{P})$ to $\mathsf{C}_{\mathrm{s}}^*(h_1, h_2, p, \mathsf{P})$,
 	we have to substitute a portion $\frac{1-p}{2}\log \left[1+h_2\mathsf{P}\right]$ of $\mathsf{C}(h_2\mathsf{P})$ by $\frac{1-p}{2}\log \left[1+h_1\mathsf{P}\right]$ and add an additional contribution $ \log e\frac{1-p}{2}\left(\frac{h_2}{1+h_2\mathsf{P}} - 		\frac{h_1}{1+h_1\mathsf{P}}\right)\mathsf{P}$. This is an improvement if
	 \begin{align}
 		\log \left[1+h_2\mathsf{P}\right] &> \log \left[1+h_1\mathsf{P}\right] + \log e\left(\frac{h_2}{1+h_2\mathsf{P}} - 		\frac{h_1}{1+h_1\mathsf{P}}\right)\mathsf{P} \\
 		\Leftrightarrow \log \left[1+h_2\mathsf{P}\right] - \log e\frac{h_2\mathsf{P}}{1+h_2\mathsf{P}} &> \log \left[1+h_1\mathsf{P}\right] - \log e \frac{h_1\mathsf{P}}{1+h_1\mathsf{P}}.
 	\end{align}
 	Since the function $f(x) = \log(1+x) - \log e \frac{x}{1+x}$ is monotonically increasing for $x \geq 0$
 	, this is equivalent to saying that the first-order term of \eqref{eq:sum_rate_ub} is always smaller than that of \eqref{eq:sum_rate_ub_hom}, since by assumption we have $h_2 \geq h_1$. The first-order advantage of \eqref{eq:sum_rate_ub} over \eqref{eq:sum_rate_ub_hom} then gets larger with increasing $h_2$. However, note that the different second-order terms make a full comparison more difficult.
 	This shows that the bound from Theorem \ref{theorem:heterogeneous_sato_outer_bound} takes the disadvantage of user 2 due to the shorter blocklength constraint into account, since it is smaller than the more naive homogeneous blocklength approach.

\section{Improved Achievability Scheme}\label{sec:achievable}
In this section, we will introduce a new achievability result for the individual power constraint using \textit{composite shell codes}, which are a new extension of the shell code concept to the early decoding technique of \cite{Lin.2021}. We derive the number of necessary symbols for a successful early decoding and the rate region using composite shell codes.

\subsection{Composite Shell Codes}
The authors of \cite{Lin.2021} derived an achievable rate region using superposition coding with IPC by the \textit{early decoding} technique with i.i.d. Gaussian codebooks. Since the dispersion of i.i.d. Gaussian codebooks is suboptimal, it is desirable to apply shell codewords to the setup. We define the set of shell codewords as
\begin{align}
	\mathcal{S}^{(n)}(\mathsf{P}) := \left\{x^n: \|x^n\|^2 = n\mathsf{P}\right\}.
\end{align}
They are characterized by the fact that they fulfill the power constraint $\|x^n\|^2 \leq  n\mathsf{P}$ with equality (\textit{equal-power property}).
The challenge is that, for early decoding, the codeword is not received completely, so the received codeword does not necessarily fulfill the equal-power property. Therefore, we introduce a modified class of shell codes to solve this problem.

We define the set of \textit{composite shell codewords} (CSC) as the set of all codewords composed of two sub-codewords that both fulfill the equal-power property, so the cost violation probability is equal to zero:
\begin{IEEEeqnarray}{rCl}
	\mathcal{S}^{(n_2, n_1)}(\mathsf{P}) &:=& \left\{x^{n_1}: \|[x_1, ..., x_{n_2}]\|^2 = n_2\mathsf{P},  \; \|[x_{n_2+1}, ..., x_{n_1}]\|^2 = (n_1-n_2)\mathsf{P}\right\}.\IEEEeqnarraynumspace
\end{IEEEeqnarray}
We will choose codewords according to the uniform distribution over $\mathcal{S}^{(n_2, n_1)}(\mathsf{P})$, which will be shown to be obtained by uniformly and independently selecting the two concatenated sub-codewords from their respective power shells.

\subsection{Improved Early Decoding}
The following result is an improvement of \cite[Theorem 1]{Lin.2021} by using composite shell codes instead of an i.i.d. Gaussian codebook for user 1.
\begin{theorem}\label{theorem:ed_symbols}
		Consider an HB-GBC with channel gains $h_1 \leq h_2$, blocklength constraints $n_1 \geq n_2$, feasible set $\mathcal{F}_{\mathrm{max}}^{n_1, n_2}(\mathsf{P}_1, \mathsf{P}_2)$ and let $n_1$ be sufficiently large. Furthermore, let $\epsilon_1$, $\epsilon_{\mathrm{SIC},1}$ and $\epsilon_{\mathrm{SIC},2}$ be the target error probabilities at user 1 and at the two SIC steps at user 2, respectively, s.t. they fulfill \eqref{eq:error_allocation}. Then successful early decoding can be performed if the following inequality holds:
	\begin{align}
		n_2 \geq \left(\frac{\sqrt{\mathsf{V}(g_2\mathsf{P}_1)}Q^{-1}(\epsilon_{\mathrm{SIC},1})}{2\mathsf{C}(g_2\mathsf{P}_1)} +\sqrt{\frac{\mathsf{V}(g_2\mathsf{P}_1)(Q^{-1}(\epsilon_{\mathrm{SIC},1}))^2}{4\mathsf{C}(g_2\mathsf{P}_1)^2} + \frac{\log (\mathsf{M}_1)}{\mathsf{C}(g_2\mathsf{P}_1)}}\right)^2, \label{eq:ed_symbols}
	\end{align}
	where $g_2 := \frac{h_2}{1+h_2\bar{\mathsf{P}}_2}$ and $\bar{\mathsf{P}}_2 := \mathsf{P}_2 - \delta, \delta > 0$.
	In that case, all message size pairs fulfilling the following inequalities:
	\begin{align}
	    \log \mathsf{M}_1 &\leq  n_1\bar{\mathsf{C}}_1 - \sqrt{n_1\bar{\mathsf{V}}_1}Q^{-1}(\epsilon_1) + O\left(1\right),\\
		\log \mathsf{M}_2 &\leq n_2\mathsf{C}(h_2\bar{\mathsf{P}}_2) - \sqrt{n_2\mathsf{V}_{\mathrm{G}}(h_2\bar{\mathsf{P}}_2)}Q^{-1}(\epsilon_{\mathrm{SIC},2}) + O\left(1\right)
	\end{align}
	are achievable, where $g_1 := \frac{h_1}{1+h_1\bar{\mathsf{P}}_2}$ and 
	\begin{align}
		\bar{\mathsf{C}}_1 &:= p\mathsf{C}(g_1\mathsf{P}_1) + \bar{p}\mathsf{C}(h_1\mathsf{P}_1), \\
		\bar{\mathsf{V}}_1 &:= p\mathsf{V}(g_1\mathsf{P}_1) + \bar{p}\mathsf{V}(h_1\mathsf{P}_1).
	\end{align}
\end{theorem}
For the proof, please refer to Appendix \ref{sec:proof_ed_symbols}. The main proof steps can be summarized as follows:
\begin{itemize}
	\item We use the proposed composite shell codes for user 1 i, s.t. the first $n_2$ symbols also fulfill the equal-power property. User 2 still uses i.i.d. Gaussian codewords, s.t. the interference he produces is still Gaussian.
	\item We evaluate the information density at user 2 for the first SIC step. User 2 receives a shell codeword from user 1 of length $n_2$ with additive white Gaussian noise. The expectation and variance of the information density are $\mathsf{C}(g_2\mathsf{P}_1)$ and $\mathsf{V}(g_2\mathsf{P}_1)$, respectively. The expectation as well as the variance are now independent of the realization of $x^{n_1}$, which is an important difference to \cite{Lin.2021}. The reason for the independence of $x^{n_1}$ is that the equal-power property allows the useful simplification $\|x^{n_2}\|^2 = n_2\mathsf{P}$ in the derivation of the information density, just like in the point-to-point case with shell codes. This simplification is not available when dealing with i.i.d. Gaussian codebooks, where the expectation and variance of the information density therefore depend on the specific realization of $x^{n_1}$.
	\item Then we proceed with similar steps as in the proof of \cite[Theorem 1]{Lin.2021}, but with the difference that we do not need to bound $0 \leq \|x^{n_2}\|^2 \leq n_1\mathsf{P}_1$ in several steps, since the expressions are already independent of $x^{n_2}$.
	\item Also, we do not have to bound $n_1 > n_2$ as in \cite{Lin.2021}, and we can solve for $n_2$ using the quadratic formula.
\end{itemize}

Compared to \cite{Lin.2021b}, we now achieve the shell dispersion $\mathsf{V}(\cdot)$ in the second-order term, instead of the i.i.d. Gaussian dispersion. The use of CSC does not lead to a decreased second-order performance compared to standard shell codes. The penalty from using CSC is captured in the constant $\tilde{\mathsf{K}}$ (Lemma \ref{lemma:radon_nikodym_concatenated} in Appendix \ref{sec:proof_ed_symbols}), which belongs to the third-order term.

\begin{remark}
	As becomes clear from the proof of Theorem \ref{theorem:ed_symbols}, the improved second-order performance compared to \cite{Lin.2021b} comes at the cost of introducing a constant $\tilde{\mathsf{K}}$ which makes the third-order term larger. This will make the second-order approach more inaccurate at smaller blocklengths. 
	
	Additionally, it is possible to extend this approach by using a concatenation of more than two codewords, for example for the $m$-user HB-GBC. If $m$ codewords are concatenated (where $m$ may be any constant, but not a function of $n$), the Radon-Nikodym derivative from Lemma \ref{lemma:radon_nikodym_concatenated} in Appendix \ref{sec:proof_ed_symbols} will be bounded by 
	\begin{align}
		\frac{dP_{Y^n}(y^n)}{dQ_{Y^n}(y^n)} \leq \mathsf{K}^m,
	\end{align}
	which will make the third-order term larger and larger with increasing $m$. Since it only considers the first- and second-order term, this behavior is not captured by the Gaussian approximation.
\end{remark}

\section{Properties of the GBC in the finite blocklength regime}\label{sec:properties}
In this section, we investigate the stochastical degradedness property of the broadcast channel and the time sharing technique in the finite blocklength regime. It will be shown that there exist important differences to the asymptotic case.

\subsection{Stochastical Degradedness}
The notion of stochastical degradedness is commonly used in the asymptotic analysis of broadcast channels. A broadcast channel is called \textit{stochastically degraded} if it has the same condition marginal distributions as a physically degraded channel \cite[p. 112]{ElGamal.2011}. The SMP (Section \ref{sec:smp}) states that two channels having the same marginals is sufficient to guarantee that they also have the same asymptotic capacity region. However, as we have shown, the SMP is not valid for non-vanishing error probabilities. Therefore, even though the definition of stochastical degradedness is still valid, we can no longer use it to transform a broadcast channel in an equivalent, physically degraded channel.

\subsection{Time Sharing}\label{sec:time_sharing}

The \textit{time sharing} technique is commonly used in achievability schemes to convexify rate regions. Time sharing between two first-order achievable rate pairs $(\tilde{\mathsf{R}}_1, \tilde{\mathsf{R}}_2)$ and $(\hat{\mathsf{R}}_1, \hat{\mathsf{R}}_2)$ is a way to achieve all rates 
\begin{align}
	\left(\alpha \tilde{\mathsf{R}}_1 + \bar{\alpha}\hat{\mathsf{R}}_1, \alpha \tilde{\mathsf{R}}_2 + \bar{\alpha}\hat{\mathsf{R}}_2\right), \qquad \alpha \in [0, 1], \label{eq:time_sharing_asymptotic}
\end{align}
where $\bar{\alpha} = 1-\alpha$. These rates are obtained by dividing the total blocklength into two sub-blocks of lengths $\alpha n$ and $(1-\alpha)n$, assuming they both are integers, and using the scheme that achieves $(\tilde{\mathsf{R}}_1, \tilde{\mathsf{R}}_2)$ in the first block and the scheme that achieves $(\hat{\mathsf{R}}_1, \hat{\mathsf{R}}_2)$ in the second block \cite[p. 85]{ElGamal.2011}. 

When applying this concept in the finite blocklength regime, this means that the two codewords are even shorter than the given blocklength constraint, which was already noted by \cite{Tan.2014} for the multiple access channel. However, while the authors of \cite{Tan.2014} only state that time sharing is not available because of the blocklength constraint, we will now investigate what happens if we allow the two codewords between which time sharing is applied to be shorter than the original codeword. It turns out that the result \eqref{eq:time_sharing_asymptotic} is not valid for finite blocklengths\footnote{For the sake of simplicity, we consider homogeneous blocklenghts. The results directly transfer to the heterogeneous blocklength case.}:
\begin{theorem}\label{theorem:time_sharing_fbl}
	Time sharing in the finite blocklength regime results in a non-convex rate region. 
\end{theorem}
\begin{proof}
Consider two second-order rate pairs
\begin{align}
	(\tilde{\mathsf{R}}_1, \tilde{\mathsf{R}}_2) = \left(\tilde{\mathsf{R}}_{\mathrm{FO},1} - \frac{\tilde{\mathsf{R}}_{\mathrm{SO},1}}{\sqrt{n}}, \tilde{\mathsf{R}}_{\mathrm{FO},2} - \frac{\tilde{\mathsf{R}}_{\mathrm{SO},2}}{\sqrt{n}}\right)
\end{align}
and 
\begin{align}
	(\hat{\mathsf{R}}_1, \hat{\mathsf{R}}_2) = \left(\hat{\mathsf{R}}_{\mathrm{FO},1} - \frac{\hat{\mathsf{R}}_{\mathrm{SO},1}}{\sqrt{n}}, \hat{\mathsf{R}}_{\mathrm{FO},2} - \frac{\hat{\mathsf{R}}_{\mathrm{SO},2}}{\sqrt{n}}\right),
\end{align}
where we use the subscripts "FO" and "SO" to denote the first-order term and the part of the second-order term which is independent of $n$.
By using time sharing as described above, i.e., using two sub-blocks of lengths $\alpha n$ and $\bar{\alpha}n$, we can achieve the rates
\begin{align}
	&\left(\alpha\tilde{\mathsf{R}}_{\mathrm{FO},1} + \bar{\alpha}\hat{\mathsf{R}}_{\mathrm{FO},1} - \frac{\alpha\tilde{\mathsf{R}}_{\mathrm{SO},1}}{\sqrt{\alpha n}} - \frac{\bar{\alpha}\hat{\mathsf{R}}_{\mathrm{SO},1}}{\sqrt{\bar{\alpha}n}}, \alpha\tilde{\mathsf{R}}_{\mathrm{FO},2} + \bar{\alpha}\hat{\mathsf{R}}_{\mathrm{FO},2} - \frac{\alpha\tilde{\mathsf{R}}_{\mathrm{SO},2}}{\sqrt{\alpha n}} - \frac{\bar{\alpha}\hat{\mathsf{R}}_{\mathrm{SO},2}}{\sqrt{\bar{\alpha}n}}\right)
\end{align}
Note that the scaling by $\alpha$ and $\bar{\alpha}$ also affects the denominators of the second-order terms, as they are dependent on the blocklength, unlike the first-order term. Consider the rates for user 1: 
\begin{align}
	\alpha\tilde{\mathsf{R}}_{\mathrm{FO},1} + \bar{\alpha}\hat{\mathsf{R}}_{\mathrm{FO},1} - \frac{\alpha\tilde{\mathsf{R}}_{\mathrm{SO},1}}{\sqrt{\alpha n}} - \frac{\bar{\alpha}\hat{\mathsf{R}}_{\mathrm{SO},1}}{\sqrt{\bar{\alpha}n}} &= \alpha\tilde{\mathsf{R}}_{\mathrm{FO},1} + \bar{\alpha}\hat{\mathsf{R}}_{\mathrm{FO},1} - \frac{\sqrt{\alpha}\tilde{\mathsf{R}}_{\mathrm{SO},1}}{\sqrt{n}} - \frac{\sqrt{\bar{\alpha}}\hat{\mathsf{R}}_{\mathrm{SO},1}}{\sqrt{n}} \\
	&\leq \alpha\tilde{\mathsf{R}}_{\mathrm{FO},1} + \bar{\alpha}\hat{\mathsf{R}}_{\mathrm{FO},1} - \frac{\alpha\tilde{\mathsf{R}}_{\mathrm{SO},1}}{\sqrt{n}} - \frac{\bar{\alpha}\hat{\mathsf{R}}_{\mathrm{SO},1}}{\sqrt{n}} \\
	&= \alpha \tilde{\mathsf{R}}_1 + \bar{\alpha}\hat{\mathsf{R}}_1,
\end{align}
where the inequality comes from the fact that $\alpha \leq 1$ and $\bar{\alpha} \leq 1$.
The same holds for user 2. Therefore, the achievable rates by time sharing for both users are smaller than the convex combination $\alpha \tilde{\mathsf{R}}_k + \bar{\alpha}\hat{\mathsf{R}}_k$ of the two second-order rates, which leads to a non-convex rate region.
\end{proof}
Theorem \ref{theorem:time_sharing_fbl} states that, in contrast to the asymptotic case, time sharing cannot guarantee that all convex combinations of two achievable second-order rate pairs are also achievable.
Fig. \ref{fig:time_sharing_vs_n} shows what time sharing between the single user second-order rates
\begin{align}
	R_{k,0} := \mathsf{C}(\mathsf{P}) - \sqrt{\frac{\mathsf{V}_{\mathrm{G}}(\mathsf{P})}{n}}, \qquad k = 1,2
\end{align}
looks like for different blocklengths in a symmetric GBC with $h_1 = h_2 = 1$, $\mathsf{P} = 10$ and $\epsilon = 10^{-6}$. All rates are normalized w.r.t. these single user rates, such that the rates for the different blocklengths are comparable more easily. We vary $\alpha$ from zero to one with step size 0.05, each point on the curves represents one value of $\alpha$. It can be observed that for finite blocklengths, time sharing results in a non-convex rate region, unlike in the asymptotic case, where the convex combination of all achievable rate pairs can be achieved. The curves deviate more from the asymptotic the smaller the blocklength is, since for smaller blocklengths, the second-order term is becoming more dominant.

\begin{figure}[ht]
	\begin{center}
		\input{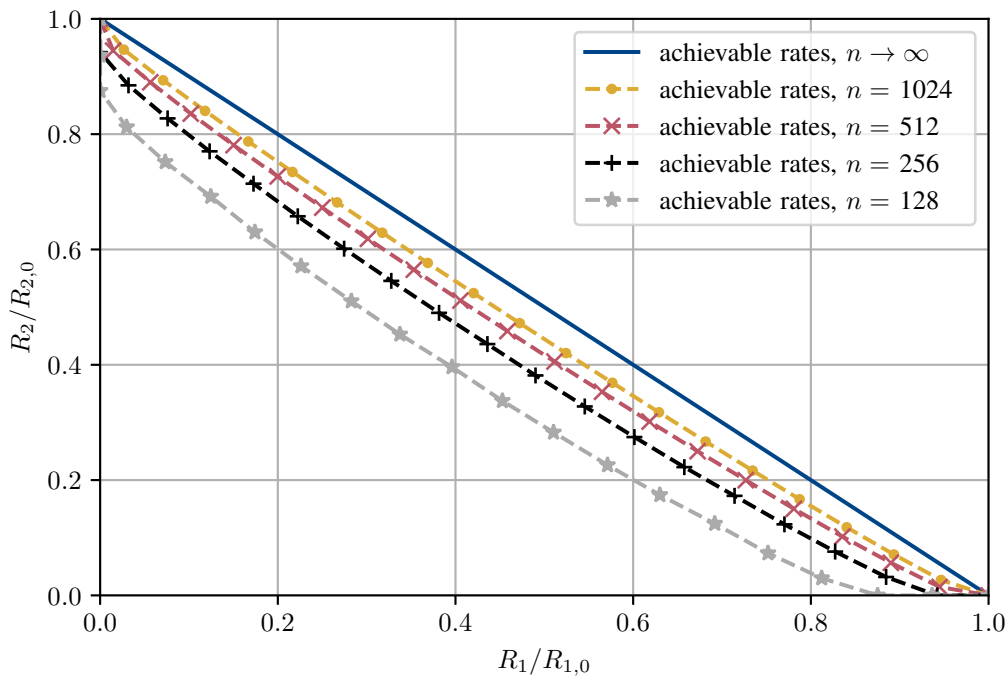}
	\end{center}
	\caption{Achievable rate region through time sharing between the single-user rates $R_{1,0}$ and $R_{2,0}$, normalized to the single user rates, at different blocklengths, for $h_1 = h_2 = 1$, $\mathsf{P} = 10$ and $\epsilon = 10^{-6}$.}
	\label{fig:time_sharing_vs_n}
\end{figure}

Note that even the second-order approach to time sharing presented in Fig. \ref{fig:time_sharing_vs_n} has its limitations. For $\alpha$ close to 1 or 0, the lengths of the considered sub-blocks are very small and can be in the order of magnitude of only a few bits. For these blocklengths, the Gaussian approximation is not very accurate and therefore, our second-order time sharing results are not very reliable for these values of $\alpha$. The most reliable results are from the region where $\alpha$ is close to 0.5, where the two sub-blocks still have a considerable blocklength.

\section{Numerical Results}\label{sec:numerical}
\subsection{Outer Bound}
In this section, we compare the sum rate upper bounds from Theorem \ref{theorem:heterogeneous_sato_outer_bound} and Corollary \ref{theorem:homogeneous_sato_outer_bound} to the achievable sum rates under the SPC that were derived in \cite{Lin.2021b, Lin.2022}. Therefore, in Fig. \ref{fig:ub_comparison_n1}, we consider a HB-GBC with sum power constraint $\mathsf{P} = 10$, $\epsilon = 2\cdot 10^{-6}$, channel gains $h_1 = 1$ and $h_2 \in \{1.5, 10\}$, fixed blocklength ratio $p=n_2/n_1=0.9$ and we vary $n_1 \in [128, 2048]$. We assume $\epsilon_1 = 10^{-6}$, $\epsilon_{\mathrm{SIC},1} = \epsilon_{\mathrm{SIC},2} = 5 \cdot 10^{-7}$ for the achievable rate expressions from \cite{Lin.2021b}.
\begin{figure}[ht]
	\begin{center}
		\scalebox{1}{\input{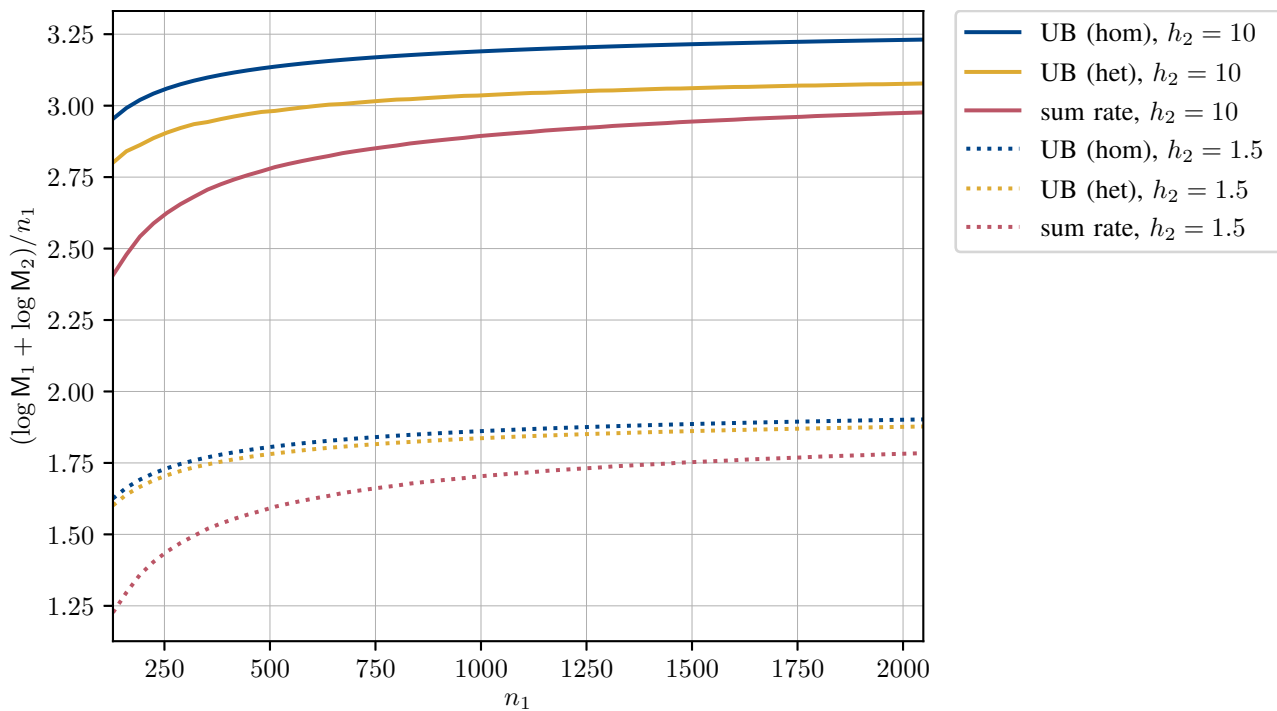}}
	\end{center}
	\caption{Comparison of the sum rate upper bounds from Theorem \ref{theorem:heterogeneous_sato_outer_bound} (het) and Corollary \ref{theorem:homogeneous_sato_outer_bound} (hom) with the achievable sum rates based on \cite{Lin.2021b} for two channels with $h_2 = 1.5$ and $h_2 = 10$.}
	\label{fig:ub_comparison_n1}
\end{figure}

We can observe that for smaller $h_2$, the homogeneous upper bound and the heterogeneous upper bound are closer to each other and the heterogeneous outer bound always outperforms the homogeneous outer bound, which confirms the observation in the end of Section \ref{sec:outer_bound}. The gap between achievable rates and outer bounds is decreasing for larger $n_1$, but does not vanish, since Sato's outer bound is not tight even in the asymptotic regime.

In Fig. \ref{fig:rate_region}, we compare the complete outer bounds from Theorem \ref{theorem:heterogeneous_sato_outer_bound} and Corollary \ref{theorem:homogeneous_sato_outer_bound} to the achievable rate regions from the achievability schemes \textit{early decoding} and \textit{hybrid NOMA} from \cite{Lin.2021b}. We keep most of the system parameters the same as those in Figure \ref{fig:ub_comparison_n1}, except $h_2 = 50$, $n_1 =1024$ and $n_2 = 840$. The maximal sum rate is achieved at $(0.22, 3.48)$ because of the channel advantage of user 2. Note that it is not possible to simply combine the two achievable schemes into a convex rate region because time sharing is not available as a tool for convexification in the finite blocklength regime (Section \ref{sec:time_sharing}). Again, the heterogeneous outer bound is tighter than the homogeneous outer bound. The gap between the achievable sum rates and the outer bound is smaller when they are close to the single user rates. The sum rate upper bound is looser than the single-user rate upper bounds due to the assumption of a cooperative receiver for the sum rate upper bound.
\begin{figure}[ht]
	\begin{center}
		\scalebox{1}{\input{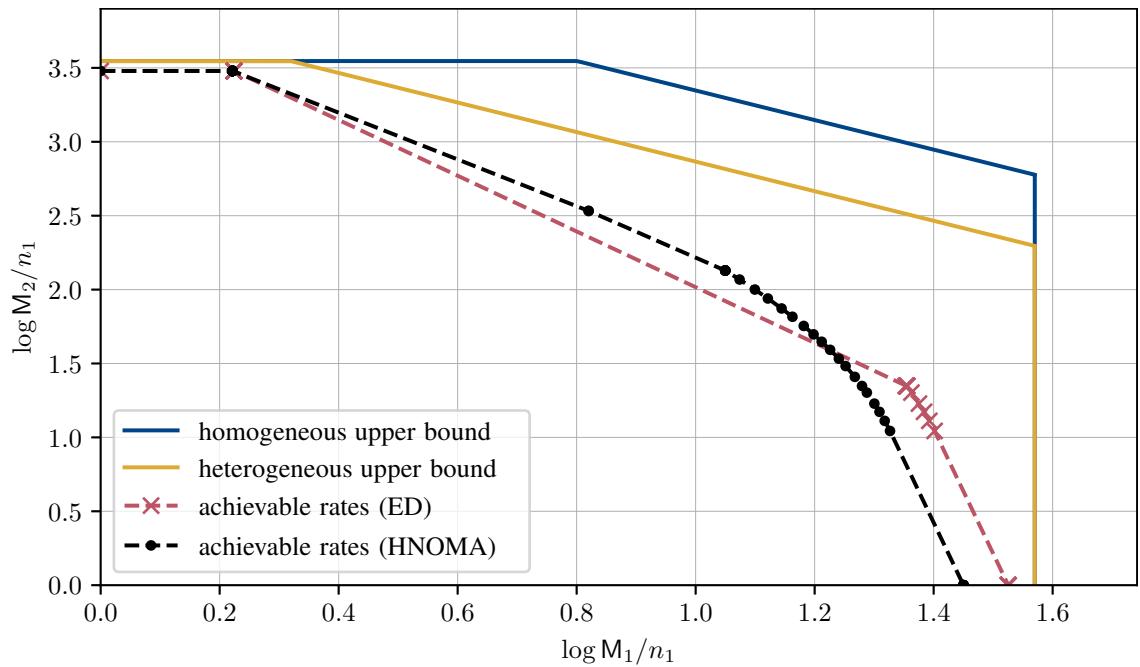}}
	\end{center}
	\caption{Achievable rate regions using Early Decoding (ED) and Hybrid NOMA (HNOMA) as in \cite{Lin.2021b} compared to our outer bounds under the SPC.}
	\label{fig:rate_region}
\end{figure}

\subsection{Latency Reduction}
We now turn to the HB-GBC with IPC and compare the number of necessary symbols $n_2$ for a successful early decoding from Theorem \ref{theorem:ed_symbols} ("Shell") to that from \cite{Lin.2021} ("i.i.d.") and the asymptotic result from \cite{Azarian.2005}, which is also considered in \cite{Lin.2021}. The considered system parameters are $\epsilon_1 = \epsilon_2 = 10^{-6}$, $h_1 = 1$, $\mathsf{P}_1 = 8$ and $\mathsf{P}_2 = 0.2$ and $n_1$ is varied. We numerically search for the best value of $\epsilon_{\mathrm{SIC},1} \in (0, \epsilon_2]$.
\begin{figure}[ht]
	\begin{center}
		\scalebox{0.8}{\input{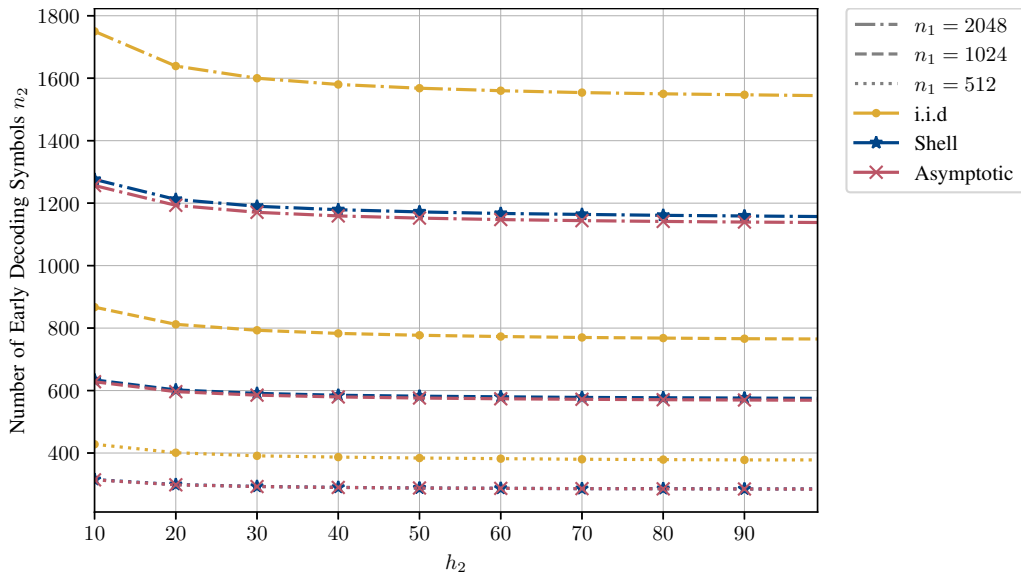}}
	\end{center}
	\caption{Number of necessary symbols for a successful early decoding at different blocklengths for the i.i.d. Gaussian inputs \cite{Lin.2021}, shell inputs for user 1 (Theorem \ref{theorem:ed_symbols}) and the asymptotic analysis \cite{Azarian.2005}.}
	\label{fig:ed_symbols}
\end{figure}

In Fig. \ref{fig:ed_symbols}, we observe that the improvement of Theorem \ref{theorem:ed_symbols} over \cite{Lin.2021} is significant. Our results are very close to the asymptotic limit for all considered blocklengths. From these results, we can conclude that early decoding is a promising technique even in the finite blocklength regime, especially when shell codes are used.

\section{Conclusion}
In this paper, we have developed a novel approach of deriving Sato-type outer bounds in the finite blocklength regime. We have applied this technique to the Gaussian broadcast channel with heterogeneous blocklength constraints using two different bounding approaches. We have also improved previous achievability results based on a technique called \textit{early decoding} by using a composite shell code for the user with the looser blocklength constraint. Numerical results show a significant improvement in terms of latency reduction over the previous early decoding result, and our results are now very close to the asymptotic limit.

\appendix
\subsection{Proof of Lemma \ref{lemma:outer_bound_lemma}}\label{sec:outer_bound_lemma_proof}
According to \eqref{eq:smp_error_prob_exact}, joint distributions with the same marginals affect the probability of a simultaneous decoding error $\hbanderror$, but not the individual error probabilities $\hbinderror{1}$ and $\hbinderror{2}$. Therefore, if we consider a new channel $V \in \mathcal{W}_{\mathrm{SM}}(W)$ with the same average system error probability constraint $\epsilon$, $\mathcal{C}_{V}(n_1, n_2, \mathcal{F}^{n_1}, \epsilon)$ may be larger or smaller than $\mathcal{C}_{W}(n_1, n_2, \mathcal{F}^{n_1}, \epsilon)$. If, e.g., the joint error probability $\hbanderrortilde$ of the new channel is smaller than $\hbanderror$, then by \eqref{eq:smp_error_prob_exact}, we have to reduce $\hbinderror{1} + \hbinderror{2}$ in order to fulfill the same error probability constraint $\epsilon$. However, reducing the individual error probabilities makes the $\mathcal{C}_{V}(n_1, n_2, \mathcal{F}^{n_1}, \epsilon)$ smaller and therefore, we cannot simply use it as an outer bound on $\mathcal{C}_{W}(n_1, n_2, \mathcal{F}^{n_1}, \epsilon)$. As a result, to derive an outer bound on $\mathcal{C}_{W}(n_1, n_2, \mathcal{F}^{n_1}, \epsilon)$, we have to ensure that both the individual and total error probabilities are at least as large as those in the original channel $W$. 

When transmitting through $W$, the individual error probabilities of the capacity-achieving code cannot be larger than $\epsilon$, since otherwise, the error probability constraint \eqref{eq:total_error_prob} would be violated. Therefore, we now fix the individual error probabilities of the new channel as $\hbinderrortilde{1} = \hbinderrortilde{2} = \epsilon$, and by \eqref{eq:smp_error_prob_exact}, we get
	\begin{align}
		\hbtoterrortilde &= \hbinderrortilde{1} + \hbinderrortilde{2} - \hbanderrortilde \geq \epsilon, \label{eq:tilde_error_lb}
	\end{align}
	since $\hbanderrortilde \leq \min\{\hbinderrortilde{1}, \hbinderrortilde{2}\} = \epsilon$. Thus, we guarantee that the average system error probability of the new channel is at least as large as that of the original channel, which results in an outer bound on $\mathcal{C}_{W}(n_1, n_2, \mathcal{F}^{n_1}, \epsilon)$. 
	
	On the other side, since $\hbanderrortilde \geq 0$, we have $\hbtoterrortilde \leq 2\epsilon.$
	Taking $\hbtoterrortilde = 2\epsilon$ gives us a uniform upper bound of $\hbtoterrortilde$ w.r.t. all possible $\hbanderrortilde$, so we can guarantee that
	\begin{align}
		\mathcal{C}_{W}(n_1, n_2, \mathcal{F}^{n_1}, \epsilon) \subseteq \mathcal{C}_{V}(n_1, n_2, \mathcal{F}^{n_1}, 2\epsilon).
	\end{align}

\subsection{Proof of Theorem \ref{theorem:heterogeneous_sato_outer_bound}}\label{sec:proof_het_sato}
The upper bounds on the individual rates directly follow from the single-user rate upper bounds \cite[Theorem 65]{Polyanskiy.2010}.

For the sum rate upper bound, we use the modified information spectrum approach presented in \cite{MolavianJazi.2014}. We consider a GBC where the two noise terms are correlated via a correlation $\rho \in [0,1]$. The modified information density for the cooperative receiver in this scenario is

\begin{IEEEeqnarray}{rCl}
    \tilde{i}(x^{n_1}; Y_1^{n_1}, Y_2^{n_2}) &=& \sum_{i=1}^{n_2} \log\frac{dP_{Y_1Y_2|X}(Y_{1,i},Y_{2,i}|x_i)}{dQ_{Y_1Y_2}(Y_{1,i},Y_{2,i})} + \sum_{i=n_2+1}^{n_1}\log\frac{dP_{Y_{1}|X}(Y_{1,i}|x_i)}{dQ_{Y_1}(Y_{1,i})}, \label{eq:id_het_bl}
\end{IEEEeqnarray}
where $P_{Y_1Y_2|X}$ is the channel distribution and $Q_{Y_1Y_2}$ is the output distribution induced by i.i.d. Gaussian inputs. Its normalized expectation is
\begin{align}
	\mathbb{E}\left[\frac{1}{n_1}\tilde{i}(x^{n_1}; Y_1^{n_1}, Y_2^{n_2})\right] = \mathsf{C}_{\rho,1} + \frac{\mathsf{C}_{\rho,2}}{n_1}\cdot\sum_{i=n_2+1}^{n_1}x_i^2,
\end{align}
where 
\begin{align}
	\mathsf{C}_{\rho,1} := p\mathsf{C}\left(h_{\rho}\mathsf{P}\right)  + \bar{p}\mathsf{C}\left(h_1\mathsf{P}\right) + \frac{\bar{p}\log e}{2}\left(\frac{h_{\rho}}{1 + h_{\rho}\mathsf{P}} - \frac{h_1}{1+h_1 \mathsf{P}}\right)\mathsf{P}\label{eq:fo_term}
\end{align}
and 
\begin{align}
	\mathsf{C}_{\rho,2} &:= \frac{\log e}{2}\left(\frac{h_1}{1+h_1\mathsf{P}} - \frac{h_{\rho}}{1 + h_{\rho}\mathsf{P}}\right). \label{eq:c_rho_2}
\end{align}
Here, we have introduced the abbreviation $h_{\rho} := (h_1+h_2-2\rho\sqrt{h_1h_2})/(1-\rho^2)$. The variance of \eqref{eq:id_het_bl} is
\begin{align}
	\mathrm{Var}\left[\frac{1}{n_1}\tilde{i}(x^{n_1}; Y_1^{n_1}, Y_2^{n_2})\right] &= \frac{\mathsf{V}_{\rho,1}}{n_1} + \frac{\mathsf{V}_{\rho,2}}{n_1^2} \cdot \sum_{i=n_2+1}^{n_1}x_i^2,
\end{align}
where 
\begin{align}
	\mathsf{V}_{\rho,1} :=& \frac{\log^2 e}{4}\left(\frac{p\cdot2h_{\rho}^2\mathsf{P}^2+ 4h_{\rho}\mathsf{P}}{(1 + h_{\rho}\mathsf{P})^2 } + \bar{p}\frac{2h_1^2\mathsf{P}^2}{(1+h_1\mathsf{P})^2} \right)
\end{align}
and 
\begin{IEEEeqnarray}{rCl}
	\mathsf{V}_{\rho,2} &:=& \frac{\log^2 e}{4}\left(\frac{4h_1}{(1+h_1\mathsf{P})^2} - \frac{4h_{\rho}}{(1+ h_{\rho}\mathsf{P})^2}\right).
\end{IEEEeqnarray}

We now employ the information spectrum converse \cite{Hayashi.2009, MolavianJazi.2014} to our setup, which states that every $(n_1, \mathsf{M}, \epsilon, \mathcal{F}^{n_1})$-code has to satisfy
\begin{align}
	\epsilon \geq \mathrm{Pr}\left[\tilde{i}(X^{n_1}; Y_1^{n_1}, Y_2^{n_2}) \leq \log \gamma_{n_1}\right] - \frac{\gamma_{n_1}}{\mathsf{M}}, \label{eq:info_spec_converse}
\end{align}
where for the cooperative receiver, we have $\mathsf{M} = \mathsf{M}_1\mathsf{M}_2$. We set  $\gamma_{n_1} = \mathsf{M}/\sqrt{n_1}$ and lower bound the probability expression in \eqref{eq:info_spec_converse} using the Berry-Esseen Theorem \cite[Theorem 44]{Polyanskiy.2010} by
\begin{IEEEeqnarray}{rl}
	&\mathrm{Pr}\left[\frac{1}{n_1}\tilde{i}(x^{n_1}; Y_1^{n_1}, Y_2^{n_2})  \leq \frac{\mathsf{M}/\sqrt{n_1}}{n_1}\right] \geq Q\Bigg(\underbrace{\frac{\mathsf{C}_{\rho,1} + \frac{\mathsf{C}_{\rho,2}}{n_1}\cdot\sum_{i=n_2+1}^{n_1}x_i^2 - \log\left(\frac{\mathsf{M}}{\sqrt{n_1}}\right)/{n_1}}{\sqrt{\mathsf{V}_{\rho,1} + \frac{\mathsf{V}_{\rho,2}}{n_1}\cdot\sum_{i=n_2+1}^{n_1}x_i^2}/{n_1}}}_{=:r_{m,\rho}(n_1)}\Bigg) - \frac{\mathsf{B}_1}{\sqrt{n_1}}. \label{eq:multi_user_info_spectrum_corr}\nonumber\\\label{eq:converse_prob_lb}
\end{IEEEeqnarray}
We observe that $\mathsf{C}_{\rho,2}$ is negative and $\mathsf{V}_{\rho,2}$ is positive if $\mathsf{P}^2 \geq \frac{1}{h_1h_{\rho}}$.
We can now upper bound $r_{\rho,m}(n_1)$ by using $\sum_{i=n_2+1}^{n_1}x_i^2 \geq 0$ in the numerator and $0 \leq \sum_{i=n_2+1}^{n_1}x_i^2 \leq (n_1-n_2)\mathsf{P}$ in the denominator:
\begin{IEEEeqnarray}{rCl}
	r_{\rho,m}(n_1) &\leq&  \frac{\mathsf{C}_{\rho,1} - \log(\mathsf{M}/\sqrt{n_1})/{n_1}}{\sqrt{\mathsf{V}_{\rho,1} + \mathsf{V}_{\rho,2}\cdot (1-p)\mathsf{P}\cdot \mathbbm{1}\left(\mathsf{P}^2 < \frac{1}{h_1h_{\rho}}\right)}}, \IEEEeqnarraynumspace \label{eq:r_ub}
\end{IEEEeqnarray}
where $\mathbbm{1}(\cdot)$ is the indicator function. Inserting \eqref{eq:r_ub} into \eqref{eq:converse_prob_lb} and using \eqref{eq:converse_prob_lb} to upper bound \eqref{eq:info_spec_converse}, we get
\begin{align}
	2\epsilon \geq Q\left(\frac{\mathsf{C}_{\rho,1} - \log(\mathsf{M}/\sqrt{n_1})/{n_1}}{\sqrt{\mathsf{V}_{\rho,1} + \mathsf{V}_{\rho,2}\cdot (1-p)\mathsf{P}\cdot \mathbbm{1}\left(\mathsf{P}^2 < \frac{1}{h_1h_{\rho}}\right)}}\right) - \frac{\bar{\mathsf{B}}}{\sqrt{n_1}}, 
\end{align}
where $\bar{\mathsf{B}} := \mathsf{B}_1 +1$.  Note that we have increased the error probability constraint from $\epsilon$ to $2\epsilon$, because this will later enable us to use Lemma \ref{lemma:outer_bound_lemma}. By solving for the number of messages, we find
\begin{align}
	&\log \mathsf{M}_1 + \log \mathsf{M}_2 \nonumber\\
	&\leq n_1\mathsf{C}_{\rho,1} - \sqrt{n_1\left(\mathsf{V}_{\rho,1} + \mathsf{V}_{\rho,2}\cdot (1-p)\mathsf{P}\cdot \mathbbm{1}\left(\mathsf{P}^2 < \frac{1}{h_1h_{\rho}}\right)\right)}Q^{-1}\left(\epsilon +\frac{\bar{\mathsf{B}}}{\sqrt{n_1}}\right) + \frac{1}{2} \log n_1 \\
	&= n_1\mathsf{C}_{\rho,1} - \sqrt{n_1\left(\mathsf{V}_{\rho,1} + \mathsf{V}_{\rho,2}\cdot (1-p)\mathsf{P}\cdot \mathbbm{1}\left(\mathsf{P}^2 < \frac{1}{h_1h_{\rho}}\right)\right)}Q^{-1}(\epsilon) + \frac{1}{2} \log n_1 + O\left(1\right), \label{eq:ub_with_rho}
\end{align}
where the last step follows from the Taylor expansion of $Q^{-1}$. 

Lemma \ref{lemma:outer_bound_lemma} now enables us to use the channel distribution (or, equivalently, the correlation $\rho$) that minimizes the first-order term \eqref{eq:fo_term} of \eqref{eq:ub_with_rho}. We consider the individual summands separately.

\begin{enumerate}
	\item Since $\mathsf{C}(\cdot)$ is a monotonically increasing function, we can minimize the argument of $\mathsf{C}(h_{\rho}\mathsf{P})$ as follows:
    	\begin{align}
    		\min_{\rho \in [0,1)} h_{\rho}\mathsf{P} = \min_{\rho \in [0,1)}\frac{h_1 + h_2-2\rho\sqrt{h_1h_2}}{1-\rho^2}\mathsf{P}. \label{eq:fo_term_minim_part1}
    	\end{align}
    Its derivative is
    \begin{align}
    	\frac{d}{d\rho} h_{\rho} = \frac{-2\sqrt{h_1h_2}\rho^2 + 2(h_1+h_2)\rho-2\sqrt{h_1h_2}}{(1-\rho^2)^2}
    \end{align}
    and the zeros of the derivatives are
    \begin{align}
    	\rho_1 = \sqrt{\frac{h_1}{h_2}}\qquad \text{and} \qquad \rho_2 = \sqrt{\frac{h_2}{h_1}}.
    \end{align}
    Since we know that $h_2 \geq h_1$, we have $\rho_2 \geq 1$, which is outside of the feasible domain for $\rho$. It can be verified that $h_{\rho_1} = h_2$ and
    \begin{align}
    	\left.\frac{d^2}{d\rho^2} h_{\rho}\right|_{\rho=\rho_1} > 0,
    \end{align}
    i.e., $\rho_1$ is indeed the minimum.
	
	\item the term 
		\begin{align}
			 \frac{1-p}{2}\log \left[1+h_1\mathsf{P}\right]
		\end{align}is independent of $\rho$,
		\item for the third term
		\begin{align}
			\log e\frac{1-p}{2}\left(\frac{h_1+h_2-2\rho\sqrt{h_1h_2}}{1-\rho^2 + \mathsf{P}(h_1+h_2-2\rho\sqrt{h_1h_2})} - \frac{h_1}{h_1\mathsf{P}+1}\right)\mathsf{P} = 0, \label{eq:third_summand_sato}
		\end{align}
		we define $h_{12,\mathsf{P}}(\rho) := \frac{h_1+h_2-2\rho\sqrt{h_1h_2}}{1-\rho^2 + \mathsf{P}(h_1+h_2-2\rho\sqrt{h_1h_2})}$, which is the only part of \eqref{eq:third_summand_sato} which is dependent on $\rho$, and find
		\begin{align}
			\frac{d}{d\rho} h_{12,\mathsf{P}}(\rho) = \frac{-2\sqrt{h_1h_2}\rho^2 + 2(h_1+h_2)\rho-2\sqrt{h_1h_2}}{\left(1-\rho^2 + \mathsf{P}(h_1+h_2-2\rho\sqrt{h_1h_2})\right)^2}
		\end{align}
	and the zeros of that derivative are 
	\begin{align}
		\rho_1 = \sqrt{\frac{h_1}{h_2}}\qquad \text{and} \qquad \rho_2 = \sqrt{\frac{h_2}{h_1}}.
	\end{align}
	Again, only $\rho_1$ is a feasible value for $\rho$ and since $\left.\frac{d^2}{d\rho^2} h_{12,\mathsf{P}}(\rho)\right|_{\rho=\rho_1} > 0$, it is a minimum.
	\end{enumerate}
	Since $\rho = \sqrt{\frac{h_1}{h_2}}$ minimizes all three terms individually, it also minimizes the first-order term \eqref{eq:fo_term} as a whole. Inserting $\rho = \sqrt{\frac{h_1}{h_2}}$ into \eqref{eq:ub_with_rho}, we obtain the values $\mathsf{C}_{\mathrm{s}}^*(h_1, h_2, p, \mathsf{P})$ and $\mathsf{V}_{\mathrm{s}}^*(h_1, h_2, p, \mathsf{P})$ from \eqref{eq:converse_sato_fo_term_opt} and \eqref{eq:converse_sato_variance_opt}.

\subsection{Proof of Theorem \ref{theorem:ed_symbols}}\label{sec:proof_ed_symbols}
\subsubsection{Number of symbols for a successful early decoding}
The proof is based on the modified Dependence Testing bound introduced in \cite[Theorem 3]{MolavianJazi.2015}, which we restate in the following for the single-user case:
\begin{theorem}\label{theorem:dt-modified}
	For a general point-to-point channel, any input distribution $P_{X^n}$ and any output distribution $Q_{Y^n}$, there exists an $(n, \mathsf{M}, \epsilon, \mathcal{F}^n)$-code that satisfies
	\begin{align}
		\epsilon \leq &\Pr \left[\tilde{i}(X^n;Y^n) \leq \log \gamma_n\right] + \mathsf{K}_n \mathsf{M}\Pr\left[\tilde{i}(X^n;\bar{Y}^n) > \log \gamma_n\right] + \Pr[X^n\notin \mathcal{F}^n]\label{eq:dt_mod}
	\end{align}
	where $P_{X^nY^n\bar{Y}^n}(a,b,c) = P_{X^n}(a)P_{Y^n|X^n}(b|a)Q_{Y^n}(c)$, $\gamma_n$ is a positive threshold and the coefficient $\mathsf{K}_n$ is defined over the Radon-Nikodym derivative
	\begin{align}
		\mathsf{K}_n := \sup_{y^n \in \mathcal{Y}^n} \frac{dP_{Y^n}(y^n)}{dQ_{Y^n}(y^n)}.
	\end{align}
\end{theorem}

First, we derive the uniform distribution on the composite power shell $\mathcal{S}^{(n_2, n_1)}(\mathsf{P})$. Choosing $x^{n_1}$ uniformly from $\mathcal{S}^{(n_2, n_1)}(\mathsf{P})$ is equivalent to first choosing the first $n_2$ symbols uniformly from $\mathcal{S}^{(n_2)}(\mathsf{P})$ and then independently choosing the remaining symbols uniformly from $\mathcal{S}^{(n_1-n_2)}(\mathsf{P})$. In the following, we use the notation $x^{n_2} := [x_1, ..., x_{n_2}]$ and $x^{n_1-n_2} := [x_{n_2+1}, ..., x_{n_1}]$. The distribution of the first $n_2$ symbols is
	\begin{align}
		P_{X^{n_2}}(x^{n_2}) = \frac{\mathbbm{1}(x^{n_2} \in \mathcal{S}^{(n_2)}(\mathsf{P}))}{\mathsf{S}_{n_2}(\sqrt{n_2\mathsf{P}})}, \label{eq:shell_distribution_firstpart}
	\end{align}
	where $\mathsf{S}_n(r):= \frac{2\pi^\frac{n}{2}}{\Gamma \left(\frac{n}{2}\right)}r^{n-1}$, and the distribution of the last $n_1 - n_2$ symbols is
	\begin{align}
		P_{X^{n_1-n_2}}(x^{n_1-n_2}) = \frac{\mathbbm{1}(x^{n_1-n_2} \in \mathcal{S}^{(n_1-n_2)}(\mathsf{P}))}{\mathsf{S}_{n_1 - n_2}(\sqrt{(n_1-n_2)\mathsf{P}})}.\label{eq:shell_distribution_secondpart}
	\end{align}
	Since the choice of the last $n_1 - n_2$ symbols is independent of the choice of the first $n_2$ symbols, the total distribution of $x^{n_1}$ is
	\begin{align}
		P_{X^n}(x^{n_1}) = P^{\mathrm{Unif}}_{\mathcal{S}^{(n_2, n_1)}(\mathsf{P})}(x^{n_1}) := \frac{\mathbbm{1}(x^{n_2}\in \mathcal{S}^{(n_2)}(\mathsf{P}))\cdot\mathbbm{1}(x^{n_1-n_2} \in \mathcal{S}^{(n_1-n_2)}(\mathsf{P}))}{\mathsf{S}_{n_2}(\sqrt{n_2\mathsf{P}})\mathsf{S}_{n_1-n_2}(\sqrt{(n_1-n_2)\mathsf{P}})},  \label{eq:concatenated_shell_distribution_complicated}
	\end{align}
	By definition of $\mathcal{S}^{(n_2, n_1)}(\mathsf{P})$, we have
	\begin{align}
		\mathbbm{1}(x^{n_1} \in \mathcal{S}^{(n_2, n_1)}(\mathsf{P})) =\mathbbm{1}(x^{n_2} \in \mathcal{S}^{(n_2)}(\mathsf{P}))\cdot\mathbbm{1}(x^{n_1-n_2} \in \mathcal{S}^{(n_1-n_2)}(\mathsf{P})), 
	\end{align}
	so we can simplify the expression \eqref{eq:concatenated_shell_distribution_complicated} as
	\begin{align}
		P^{\mathrm{Unif}}_{\mathcal{S}^{(n_2, n_1)}(\mathsf{P})}(x^{n_1}) = \frac{\mathbbm{1}(x^{n_1} \in \mathcal{S}^{(n_2, n_1)}(\mathsf{P}))}{\mathsf{S}_{n_2}(\sqrt{n_2\mathsf{P}})\mathsf{S}_{n_1-n_2}(\sqrt{(n_1-n_2)\mathsf{P}})}. \label{eq:unif_power_shell}
	\end{align}

Now, we can bound the Radon-Nikodym derivative that results from the change of measure from $P^{\mathrm{Unif}}_{\mathcal{S}^{(n_2, n_1)}(\mathsf{P})}(x^{n_1})$ to the i.i.d. Gaussian output distribution.
	
\begin{lemma}\label{lemma:radon_nikodym_concatenated}
	Let $P_{Y^{n_1}}(y^{n_1})$ by the output distribution induced by $P^{\mathrm{Unif}}_{\mathcal{S}^{(n_2, n_1)}(\mathsf{P})}(x^{n_1})$ defined in \eqref{eq:unif_power_shell} and let $Q_{Y^{n_1}}(y^{n_1}) = \mathcal{N}(y^{n_1}; \mathbf{0}^{n_1}, (1+\mathsf{P})\cdot \unitmatrix{n_1})$. Then, for $n_1$ sufficiently large,
	\begin{align}
		\frac{dP_{Y^{n_1}}(y^{n_1})}{dQ_{Y^{n_1}}(y^{n_1})} \leq \tilde{\mathsf{K}} := 729\cdot \frac{\pi}{8}\cdot\frac{(1+\mathsf{P})^2}{1+2\mathsf{P}}.
	\end{align}
\end{lemma}
For the proof of Lemma \ref{lemma:radon_nikodym_concatenated}, please refer to Appendix \ref{sec:rn-lemma-proof}. We can now analyze the performance of the coding scheme:

\begin{enumerate}
    \item Codebook generation:
For every message $m_1 \in \mathcal{M}_1$, generate a codeword according to the uniform distribution on the composite power shell:
\begin{IEEEeqnarray}{rCl}
	X_1^{n_1}(m_1) \sim P^{\mathrm{Unif}}_{\mathcal{S}^{(n_2, n_1)}(\mathsf{P}_1)}.
\end{IEEEeqnarray}
For every message $m_2 \in \mathcal{M}_2$, generate a codeword according to the i.i.d. Gaussian distribution:
\begin{IEEEeqnarray}{rCl}
	X_{2,i}(m_2) \sim \mathcal{N}(0, \bar{\mathsf{P}}_2), \qquad i = 1, ..., n_2,
\end{IEEEeqnarray}
where $\bar{\mathsf{P}}_2 = \mathsf{P}_2 - \delta$.
\item Encoding: 
Superposition coding: to send the message pair $(m_1, m_2)$, transmit $x^{n_1} = x_1^{n_1}(m_1) + [x_2^{n_2}(m_2), \mathbf{0}^{n_1-n_2}]$.

\item Decoding: In the following, we use the notation $x_1^{n_2} = [x_{1,1}, ..., x_{1, n_2}]$ and $x_2^{n_1} = [x_2^{n_2}, \mathbf{0}^{n_1-n_2}]$. User $k$ receives the signal
\begin{IEEEeqnarray}{rCl}
	y_k^{n_k} &=& \sqrt{h_k}x_1^{n_k} + \sqrt{h_k}x_2^{n_k} + z_k^{n_k}.
\end{IEEEeqnarray}
User 2 performs successive interference cancellation (SIC), while user 1 treats the interference as noise. Both decoders use the threshold decoding technique from the modified DT bound (Theorem \ref{theorem:dt-modified}), where we use the threshold $\gamma_{n_1} = \mathsf{K}_{n_1} \mathsf{M}_1$ for user 1, $\gamma_{n_2} = \mathsf{K}_{n_2} \mathsf{M}_1$ for user 2 in the first SIC phase, and $\gamma_{n_2} = \mathsf{K}_{n_2} \mathsf{M}_2$ for user 2 in the second SIC phase. 
\begin{itemize}
	\item User 1 treats user 2's codeword as noise and decodes $\hat{m}_1$ as the smallest message s.t. $i(x_1^{n_1}(\hat{m}_1), y_1^{n_1}) \geq \log (\mathsf{K}_{n_1} \mathsf{M}_1)$. If there is no such message, an error is declared. 
	\item User 2 uses successive interference cancellation (SIC) with early decoding.
	\begin{enumerate}
		\item User 2 treats his own codeword as noise and decodes $\hat{m}_1$ as the smallest message s.t. $i\left(x_1^{n_2}(\hat{m}_1), y_2^{n_2}\right) \geq \log (\mathsf{K}_{n_2} \mathsf{M}_1)$ using early decoding. If there is no such message, an error is declared. 
		\item User 2 subtracts $x_1^{n_2}(\hat{m}_1)$ from $y_2^{n_2}$ to obtain
		\begin{IEEEeqnarray}{rCl}
			\tilde{y}_2^{n_2} &=& \sqrt{h_2}x_2^{n_2} + z_2^{n_2}.
		\end{IEEEeqnarray}
		Then, he decodes $\hat{m}_2$ as the smallest message s.t. $i\left(x_2^{n_2}(\hat{m}_2), \tilde{y}_2^{n_2}\right) \geq \log (\mathsf{K}_{n_2} \mathsf{M}_2)$.
	\end{enumerate}
\end{itemize}

\item Error analysis: For Theorem \ref{theorem:ed_symbols}, only the first SIC step at user 2 is of interest. Since we treat user 2's codeword as noise, we have the effective channel
\begin{IEEEeqnarray}{rCl}
	Y_2^{n_2} &=& \sqrt{h_2}X_1^{n_2} + (\sqrt{h_2}X_2^{n_2} + Z_2^{n_2}),
\end{IEEEeqnarray}
with the noise
\begin{IEEEeqnarray}{rCl}
 	\sqrt{h_2}X_{2,i} + Z_{2,i} &\sim& \mathcal{N}(0, 1+ h_2\bar{\mathsf{P}}_2), \qquad i = 1,...n_2.
\end{IEEEeqnarray}
By normalizing the noise power to 1, we get
\begin{IEEEeqnarray}{rCl}
	Y_2^{n_2} &=& \sqrt{g_2}X_1^{n_2} + \tilde{Z}_2^{n_2},
\end{IEEEeqnarray}
where $\tilde{Z}_{2,i} \sim \mathcal{N}(0,1)$ is i.i.d., $\tilde{Z}_2^{n_2} \indep X^{n_2}$ and we define
\begin{IEEEeqnarray}{rCl}
	g_2 &:=& \frac{h_2}{1+ h_2\bar{\mathsf{P}}_2}.
\end{IEEEeqnarray}
Since $X_1^{n_1}$ is a composite shell codeword, $X_1^{n_2}$ fulfills the equal-power property. From Lemma \ref{lemma:radon_nikodym_concatenated}, we can conclude that $\mathsf{K}_{n_2} = \tilde{\mathsf{K}}$. We use the modified information density 
\begin{align}
	\tilde{i}(X^n; Y^n) := \log \frac{dP_{Y^n|X^n}(Y^n|X^n)}{dQ_{Y^n}(Y^n)}, \label{eq:modified_information_density}
\end{align}
where $Q_{Y^n}(Y^n)$ is the i.i.d. Gaussian output distribution. The expectation and variance of the conditional information density are as in the standard single-user case:
\begin{IEEEeqnarray}{rCl}
	\mathbb{E}\left[\frac{1}{n_2}\tilde{i}(x_1^{n_2}; Y_2^{n_2})\right] &=& \mathsf{C}(g_2\mathsf{P}_1), \\
	\mathrm{Var}\left[\frac{1}{n_2}\tilde{i}(x_1^{n_2}; Y_2^{n_2})\right] &=& \frac{\mathsf{V}(g_2\mathsf{P}_1)}{n_2}.
\end{IEEEeqnarray}

Then, we can bound the outage and confusion probabilities of the DT bound (Theorem \ref{theorem:dt-modified}):
\begin{IEEEeqnarray}{rCl}
	\Pr\left[\tilde{i}(x^{n_2}; Y^{n_2}) \leq \log \gamma_{n_2}\right] &\leq& Q\left(\frac{ n\mathsf{C}(g_2\mathsf{P}_1) - \log \gamma_{n_2}}{\sqrt{n_2\mathsf{V}(g_2\mathsf{P}_1)}}\right) + 6\frac{\mathsf{B}_{1}(g_2\mathsf{P}_1)}{\sqrt{n_2}}, \\
	\Pr\left[\tilde{i}(x^{n_2};\bar{Y}^{n_2}) > \log \gamma_{n_2}\right] &\leq& \frac{\mathsf{B}_{2}(g_2\mathsf{P}_1)}{\sqrt{n_2}\gamma_{n_2}},
\end{IEEEeqnarray}
and therefore, we get with $\gamma_{n_2} = \tilde{\mathsf{K}}\mathsf{M}$,
\begin{IEEEeqnarray}{rCl}
	\epsilon_{\mathrm{SIC},1}^{(n_2)} &\leq& Q\left(\frac{ n_2\mathsf{C}(g_2\mathsf{P}_1) - \log (\tilde{\mathsf{K}} \mathsf{M}_1)}{\sqrt{n_2\mathsf{V}(\mathsf{P}_1)}}\right) + 6\frac{\mathsf{B}_{1}(g_2\mathsf{P}_1)}{\sqrt{n_2}} + \frac{\mathsf{B}_{2}(g_2\mathsf{P}_1)}{\sqrt{n_2}}\\
	&=& Q\left(\frac{ n_2\mathsf{C}(g_2\mathsf{P}_1) - \log (\tilde{\mathsf{K}} \mathsf{M}_1)}{\sqrt{n_2\mathsf{V}(g_2\mathsf{P}_1)}}\right) + \frac{\mathsf{B}(\mathsf{P}_1)}{\sqrt{n_2}},
\end{IEEEeqnarray}
where we have defined $\mathsf{B}(g_2\mathsf{P}_1) = 6\mathsf{B}_1(g_2\mathsf{P}_1) + \mathsf{B}_2(g_2\mathsf{P}_1).$ By rearranging, we find
\begin{IEEEeqnarray}{rCl}
	\frac{ n_2\mathsf{C}(g_2\mathsf{P}_1) - \log (\mathsf{M}_1)}{\sqrt{n_2\mathsf{V}(g_2\mathsf{P}_1)}} &\geq& Q^{-1}\left(\epsilon_{\mathrm{SIC},1}^{(n_2)} - \frac{\mathsf{B}(g_2\mathsf{P}_1)}{\sqrt{n_2}}\right) + \frac{\log (\tilde{\mathsf{K}})}{\sqrt{n_2\mathsf{V}(g_2\mathsf{P}_1)}}\\
	&=& Q^{-1}\left(\epsilon_{\mathrm{SIC},1}^{(n_2)}\right) + O\left(\frac{1}{\sqrt{n_2}}\right),
\end{IEEEeqnarray}
where the last step follows from the Taylor expansion of the $Q$-function. By dropping the $O(1/\sqrt{n_2})$ term and substituting $m = \sqrt{n_2}$, we get the quadratic equation
\begin{IEEEeqnarray}{rCl}
	m^2 - \frac{\sqrt{\mathsf{V}(g_2\mathsf{P}_1)}Q^{-1}\left(\epsilon_{\mathrm{SIC},1}^{(n_2)}\right)}{\mathsf{C}(g_2\mathsf{P}_1)}m - \frac{\log (\mathsf{M}_1)}{{\mathsf{C}(g_2\mathsf{P}_1)}} &\geq& 0, \label{eq:n_ed_quadratic}
\end{IEEEeqnarray}
which has the roots
\begin{IEEEeqnarray}{rCl}
	m_{1,2} &=& \frac{\sqrt{\mathsf{V}(g_2\mathsf{P}_1)}Q^{-1}\left(\epsilon_{\mathrm{SIC},1}^{(n_2)}\right)}{2\mathsf{C}(g_2\mathsf{P}_1)} \pm \sqrt{\frac{\mathsf{V}(g_2\mathsf{P}_1)Q^{-1}\left(\epsilon_{\mathrm{SIC},1}^{(n_2)}\right)^2}{4\mathsf{C}^2(g_2\mathsf{P}_1)}+\frac{\log (\mathsf{M}_1)}{{\mathsf{C}(g_2\mathsf{P}_1)}}},
\end{IEEEeqnarray}
where only the solution with the positive sign in front of the square root is positive. Also, note that the polynomial in \eqref{eq:n_ed_quadratic} has a positive sign in front of the leading term, which means that the inequality \eqref{eq:n_ed_quadratic} is fulfilled for all $n_2$ greater than the second root. Therefore, early decoding is possible for 
\begin{IEEEeqnarray}{rCl}
	n_2 &\geq& \left( \frac{\sqrt{\mathsf{V}(g_2\mathsf{P}_1)}Q^{-1}\left(\epsilon_{\mathrm{SIC},1}^{(n_2)}\right)}{2\mathsf{C}(g_2\mathsf{P}_1)} + \sqrt{\frac{\mathsf{V}(\mathsf{P}_1)Q^{-1}\left(\epsilon_{\mathrm{SIC},1}^{(n_2)}\right)^2}{4\mathsf{C}^2(g_2\mathsf{P}_1)}+\frac{\log (\mathsf{M}_1)}{{\mathsf{C}(g_2\mathsf{P}_1)}}}\right)^2.
\end{IEEEeqnarray}
\end{enumerate}

\subsubsection{Achievable Rates}\label{sec:early_decoding_rates_ipc_improved}
First, we analyze the rates at receiver 1. When treating user 2's signal as noise, the effective channel is
\begin{IEEEeqnarray}{rCl}
	Y_1^{n_1} &=& \sqrt{h_1}X_1^{n_1} + (\sqrt{h_1}X_2^{n_1} + Z_1^{n_1}) = \sqrt{h_1}X_1^{n_1} + \tilde{Z}_1^{n_1},
\end{IEEEeqnarray}
where
\begin{IEEEeqnarray}{rCl}
	 \tilde{Z}_{1,i} &\sim& \mathcal{N}(0, \sigma_i), \\
	 \sigma_i &:=& \begin{cases}
	 	1+ h_1\bar{\mathsf{P}}_2, & i = 1, ..., n_2,\\
	 	1, & i = n_2+1, ..., n_1.
	 \end{cases}
\end{IEEEeqnarray}
As the reference in the modified information density, we use the Gaussian output distribution 
\begin{IEEEeqnarray}{rCl}
	Y_{1,i} &\sim& \mathcal{N}\left(0, h_1\mathsf{P}_1 + \sigma_i\right).
\end{IEEEeqnarray}
In \cite[eq. (109)]{Lin.2021b}, the conditional modified information density for this scenario has been derived as 
\begin{IEEEeqnarray}{rCrl}
	\tilde{i}\left(x^{n_1}, Y^{n_1}\right) &=& n_1 \bar{\mathsf{C}}_1 &+ \frac{\log e}{2(1+h_1(\mathsf{P}_1+\bar{\mathsf{P}}_2))}\left[\sum_{i = 1}^{n_2}(h_1x_{1,i}^2 + 2 \sqrt{h_1}x_{1,i}\tilde{Z}_{1,i}) - g_1\mathsf{P}_1\sum_{i=1}^{n_2}\tilde{Z}_{1,i}^2\right]\nonumber\\*
	&&& +\frac{\log e}{2(1+h_1\mathsf{P}_1)}\left[\sum_{i = n_2+1}^{n_1}(h_1x_{1,i}^2 + 2 \sqrt{h_1}x_{1,i}\tilde{Z}_{1,i}) - h_1\mathsf{P}_1\sum_{i=n_2+1}^{n_1}\tilde{Z}_{1,i}^2\right], \nonumber\\\label{eq:info_density_rx_1}
\end{IEEEeqnarray}
where $g_1 = \frac{h_1}{1+h_1\bar{\mathsf{P}}_2}$ and
\begin{IEEEeqnarray}{rCl}
	\bar{\mathsf{C}}_1 &:=& p\mathsf{C}(g_1\bar{\mathsf{P}}_1) + (1-p)\mathsf{C}(h_1\bar{\mathsf{P}}_1).
\end{IEEEeqnarray}
In contrast to \cite{Lin.2021b}, we use composite shell codewords instead of i.i.d. Gaussian codewords, so we can further simplify \eqref{eq:info_density_rx_1} as
\begin{IEEEeqnarray}{rCrl}
	\tilde{i}\left(x^{n_1}, Y^{n_1}\right) &=& n_1 \bar{\mathsf{C}}_1 &+ \frac{\log e}{2(1+h_1(\mathsf{P}_1+\bar{\mathsf{P}}_2))}\sum_{i = 1}^{n_2}\left[h_1\mathsf{P}_1+ 2 \sqrt{h_1}x_{1,i}\tilde{Z}_{1,i} - g_1\mathsf{P}_1\tilde{Z}_{1,i}^2\right]\nonumber\\
	&&& +\frac{\log e}{2(1+h_1\mathsf{P}_1)}\sum_{i = n_2+1}^{n_1}\left[h_1\mathsf{P}_1+ 2\sqrt{h_1}x_{1,i}\tilde{Z}_{1,i} - h_1\mathsf{P}_1\tilde{Z}_{1,i}^2\right]. \label{eq:rx1_info_density}
\end{IEEEeqnarray}
 The expectation of \eqref{eq:rx1_info_density} is
\begin{IEEEeqnarray}{rCl}
	\mathbb{E}\left[\tilde{i}\left(x^{n_1}, Y^{n_1}\right)\right]  &=& n_1\bar{\mathsf{C}}_1
\end{IEEEeqnarray}
and the variance is
\begin{IEEEeqnarray}{rCl}
	\mathrm{Var}\left[\tilde{i}\left(x^{n_1}, Y^{n_1}\right)\right] 
	&=& n_2 \mathsf{V}(g_1\mathsf{P}_1) + (n_1-n_2)\mathsf{V}(h_1\mathsf{P}_1).
\end{IEEEeqnarray}
Now, we can use the usual steps to bound the outage and confusion probability in the DT bound and get the rates from Theorem \ref{theorem:ed_symbols}.

The rate for user 2 is the standard single-user rate using i.i.d. Gaussian codebooks, since after the successful SIC, user 2 is interference free.

\subsection{Proof of Lemma \ref{lemma:radon_nikodym_concatenated}}\label{sec:rn-lemma-proof}
In \cite[Appendix B]{MolavianJazi.2015}, it is shown that the output distribution induced by the uniform input distribution on the power shell can be upper bound by
\begin{IEEEeqnarray}{rCl}
	P_{Y^n}(y^n) &\leq& P_{\mathrm{out,shell}}(n, t) := \frac{c}{2}\frac{1 + \sqrt{1+4\mathsf{P}t}}{\sqrt[4]{1+4\mathsf{P}t}}\pi^{-\frac{n}{2}}e^{-n\frac{(1+\mathsf{P})}{2}}e^{-n\frac{t}{2}}e^{n\frac{\sqrt{1+4\mathsf{P}t}}{2}}\left(1 + \sqrt{1+4\mathsf{P}t}\right)^{-\frac{n}{2}}, \nonumber\\
\end{IEEEeqnarray}
where $t := \frac{\|y^n\|^2}{n}$ and $c = 27\sqrt{\frac{\pi}{8}}$. Since in our case, the input consists of the concatenation of two independent shell codewords, and since the channel is memoryless, the output distribution in our case will be the product of the output distributions of the two individual power shell codewords. We define 
\begin{IEEEeqnarray}{c}
	t_{\firstpart} := \frac{\|[y_1, ..., y_{n_2}]\|^2}{n_2},  \qquad t_{\secondpart} := \frac{\|[y_{n_2+1}, ..., y_{n_1}]\|^2}{n_1-n_2}, \qquad \text{and} \qquad t := \frac{\|y^{n_1}\|^2}{n_1}. \IEEEeqnarraynumspace
\end{IEEEeqnarray}
Note that $pt_{\firstpart} + (1-p)t_{\secondpart} = t$, where $p = \frac{n_2}{n_1}$. The output distribution for the composite shell code is
\begin{IEEEeqnarray}{rCl}
	P_{Y^{n_1}}(y^{n_1}) &\leq& P_{\mathrm{out,shell}}(n_2, t_{\firstpart})\cdot P_{\mathrm{out,shell}}(n_1-n_2, t_{\secondpart})\\
	&=&\frac{c^2}{4}
	\frac{1 + \sqrt{1+4\mathsf{P}t_{\firstpart}}}{\sqrt[4]{1+4\mathsf{P}t_{\firstpart}}}
	\frac{1 + \sqrt{1+4\mathsf{P}t_{\secondpart}}}{\sqrt[4]{1+4\mathsf{P}t_{\secondpart}}}
	\pi^{-\frac{n_2}{2}}\pi^{-\frac{n_1-n_2}{2}}
	e^{-n_2\frac{(1+\mathsf{P})}{2}}e^{-(n_1-n_2)\frac{(1+\mathsf{P})}{2}}\nonumber\\
	&&\cdot e^{-n_2\frac{t_{\firstpart}}{2}}e^{-(n_1-n_2)\frac{t_\secondpart}{2}}
	e^{n_2\frac{\sqrt{1+4\mathsf{P}t_{\firstpart}}}{2}}e^{(n_1-n_2)\frac{\sqrt{1+4\mathsf{P}t_{\secondpart}}}{2}} \nonumber\\
	&&\cdot\left(1 + \sqrt{1+4\mathsf{P}t_{\firstpart}}\right)^{-\frac{n_2}{2}}\left(1 + \sqrt{1+4\mathsf{P}t_{\secondpart}}\right)^{-\frac{n_1-n_2}{2}} \\
	&=& \mathsf{K}_1(\mathsf{P}, t_{\firstpart}, t_{\secondpart}) \cdot
	\pi^{-\frac{n_1}{2}}
	e^{-n_1\frac{(1+\mathsf{P})}{2}}
	e^{-n_1\frac{pt_{\firstpart}+ (1-p)t_\secondpart}{2}}
	e^{n_1\left(p\frac{\sqrt{1+4\mathsf{P}t_{\firstpart}}}{2}+(1-p)\frac{\sqrt{1+4\mathsf{P}t_{\secondpart}}}{2}\right)} \nonumber\\
	&&\cdot e^{-\frac{n_1}{2}\left( p\cdot \ln\left(1 + \sqrt{1+4\mathsf{P}t_{\firstpart}}\right) + (1-p)\cdot \ln\left(1 + \sqrt{1+4\mathsf{P}t_{\secondpart}}\right)\right)} \\
	&=& \mathsf{K}_1(\mathsf{P}, t_{\firstpart}, t_{\secondpart}) \cdot
	\pi^{-\frac{n_1}{2}}
	e^{-\frac{n_1}{2}f_1(\mathsf{P}, t_{\firstpart}, t_{\secondpart}, p)},\label{eq:shell_output_approx_concatenated}
\end{IEEEeqnarray}
where we have defined 
\begin{IEEEeqnarray}{rCl}
	\mathsf{K}_1(\mathsf{P}, t_{\firstpart}, t_{\secondpart}) &:=& \frac{c^2}{4}
	\frac{1 + \sqrt{1+4\mathsf{P}t_{\firstpart}}}{\sqrt[4]{1+4\mathsf{P}t_{\firstpart}}}
	\frac{1 + \sqrt{1+4\mathsf{P}t_{\secondpart}}}{\sqrt[4]{1+4\mathsf{P}t_{\secondpart}}}
\end{IEEEeqnarray}
and
\begin{IEEEeqnarray}{rCl}
	f_1(\mathsf{P}, t_{\firstpart}, t_{\secondpart}, p) &:=& (1+\mathsf{P}) + pt_{\firstpart}+ (1-p)t_\secondpart - p\sqrt{1+4\mathsf{P}t_{\firstpart}} - (1-p)\sqrt{1+4\mathsf{P}t_{\secondpart}} \nonumber\\*
	&&+  p\cdot \ln\left(1 + \sqrt{1+4\mathsf{P}t_{\firstpart}}\right) + (1-p)\cdot \ln\left(1 + \sqrt{1+4\mathsf{P}t_{\secondpart}}\right).
\end{IEEEeqnarray}

The output distribution induced by i.i.d. Gaussian inputs is
\begin{IEEEeqnarray}{rCl}
	Q_{Y^{n_1}}(y^{n_1}) &=& (2\pi)^{-\frac{n_1}{2}} (1+\mathsf{P})^{-\frac{n_1}{2}} e^{-n_1\frac{t}{2(1+\mathsf{P})}} \nonumber \\
	&=& \pi^{-\frac{n_1}{2}} e^{-\frac{n_1}{2}f_2(\mathsf{P}, t)}, \label{eq:gaussian_output}
\end{IEEEeqnarray}
where
\begin{IEEEeqnarray}{rCl}
	f_2(\mathsf{P}, t) &:=& \ln\left(2(1+\mathsf{P})\right) + \frac{t}{(1+\mathsf{P})}.
\end{IEEEeqnarray}
We obtain the ratio of output distributions by dividing \eqref{eq:shell_output_approx_concatenated} by \eqref{eq:gaussian_output}:
\begin{IEEEeqnarray}{rCl}
	\frac{dP_{Y^{n_1}}(y^{n_1})}{dQ_{Y^{n_1}}(y^{n_1})} &\leq& \mathsf{K}_1(\mathsf{P}, t_{\firstpart}, t_{\secondpart}) \cdot e^{-\frac{n_1}{2}\left(f_1(\mathsf{P}, t_{\firstpart}, t_{\secondpart}, p) - f_2(\mathsf{P}, t)\right)}. \label{eq:output_ratio}
\end{IEEEeqnarray}

Since we are looking for an upper bound of the ratio \eqref{eq:output_ratio} for the worst case, we are finding the realization of $y^{n_1}$ (or, equivalently, $t_{\firstpart}$ and  $t_{\secondpart}$) that maximizes \eqref{eq:output_ratio}.

First, we investigate the exponent in \eqref{eq:output_ratio}, which we want to minimize in order to maximize \eqref{eq:output_ratio}:
\begin{IEEEeqnarray}{rCl}
	f(\mathsf{P}, t_{\firstpart}, t_{\secondpart}, p) &:=& f_1(\mathsf{P}, t_{\firstpart}, t_{\secondpart}, p) - f_2(\mathsf{P}, t) \\
	&=& (1+\mathsf{P}) + pt_{\firstpart}+ (1-p)t_\secondpart - p\sqrt{1+4\mathsf{P}t_{\firstpart}}-(1-p)\sqrt{1+4\mathsf{P}t_{\secondpart}} \nonumber\\
	&&+  p\cdot \ln\left(1 + \sqrt{1+4\mathsf{P}t_{\firstpart}}\right) + (1-p)\cdot \ln\left(1 + \sqrt{1+4\mathsf{P}t_{\secondpart}}\right) \nonumber\\
	&&- \ln\left(2(1+\mathsf{P})\right) - \frac{t}{(1+\mathsf{P})} \\
	&=& (1+\mathsf{P}) - \ln\left(2(1+\mathsf{P})\right) + \frac{(pt_{\firstpart}+ (1-p)t_\secondpart)\mathsf{P}}{(1+\mathsf{P})} - p\sqrt{1+4\mathsf{P}t_{\firstpart}} \nonumber\\
	&&-(1-p)\sqrt{1+4\mathsf{P}t_{\secondpart}} +  p\cdot \ln\left(1 + \sqrt{1+4\mathsf{P}t_{\firstpart}}\right) \nonumber\\
	&&+ (1-p)\cdot \ln\left(1 + \sqrt{1+4\mathsf{P}t_{\secondpart}}\right), \label{eq:exponent_function_rn}
\end{IEEEeqnarray}
where the last step follows from $t = pt_{\firstpart} + (1-p)t_{\secondpart}$. Its derivative w.r.t. $t_{\firstpart}$ is
\begin{IEEEeqnarray}{rCl}
	\frac{df(\mathsf{P}, t_{\firstpart}, t_{\secondpart}, p)}{dt_{\firstpart}} &=& 
	\frac{p\mathsf{P}}{(1+\mathsf{P})} -
	\frac{p\cdot4\mathsf{P}}{2\sqrt{1+4\mathsf{P}t_{\firstpart}}} +  \frac{p\cdot4\mathsf{P}}{(2\sqrt{1+4\mathsf{P}t_{\firstpart}})(1 + \sqrt{1+4\mathsf{P}t_{\firstpart}})},
\end{IEEEeqnarray}
which has exactly one zero at $t_{\firstpart} = 1+\mathsf{P}$. It can also be shown that 
\begin{IEEEeqnarray}{rCl}
 	\left.\frac{d^2f(\mathsf{P}, t_{\firstpart}, t_{\secondpart}, p)}{dt_{\firstpart}^2}\right|_{t_{\firstpart} = 1+\mathsf{P}} &>& 0,
\end{IEEEeqnarray}
which means that $t_{\firstpart} = 1+\mathsf{P}$ is the minimizes $f(\mathsf{P}, t_{\firstpart}, t_{\secondpart}, p)$. Since $f(\mathsf{P}, t_{\firstpart}, t_{\secondpart}, p)$ is symmetric in its arguments $t_{\firstpart}$ and $t_{\secondpart}$, the global minimum w.r.t. $t_{\secondpart}$ is also $t_{\secondpart} = 1+\mathsf{P}$. Inserting these values of $t_{\firstpart}$ and $t_{\secondpart}$ into \eqref{eq:exponent_function_rn} yields 
\begin{IEEEeqnarray}{rCl}
	f(\mathsf{P}, t_{\firstpart} = 1+\mathsf{P}, t_{\secondpart} = 1+\mathsf{P}, p) &=& 0.
\end{IEEEeqnarray}
For $n_1$ sufficiently large
, this leads to 
\begin{IEEEeqnarray}{rCl}
	\frac{dP_{Y^{n_1}}(y^{n_1})}{dQ_{Y^{n_1}}(y^{n_1})} &\leq& \mathsf{K}_1(\mathsf{P}, t_{\firstpart} = 1+\mathsf{P}, t_{\secondpart} = 1+\mathsf{P}) \\
	&=& c^2\frac{(1+\mathsf{P})^2}{1+2\mathsf{P}}.
\end{IEEEeqnarray}

\bibliographystyle{IEEEtran}
\bibliography{references}

\end{document}